\documentclass[journal]{IEEEtran}
\usepackage{amssymb}
\usepackage{amsmath}
\usepackage{amsbsy}
\usepackage{graphics,graphicx}
\usepackage{pict2e}
\usepackage[us]{datetime}
\usepackage{comment}
\usepackage{bm}
\usepackage{xcolor}

\begin{titlepage}
\title{Optimal algorithms for universal random number generation
from finite memory sources}
\author{Gadiel~Seroussi,~\IEEEmembership{Fellow,~IEEE,}
and~Marcelo~J.\ Weinberger~\IEEEmembership{Fellow,~IEEE}
\thanks{Gadiel~Seroussi is with Universidad de la Rep\'{u}blica, Montevideo,
Uruguay (e-mail: gseroussi@ieee.org).}
\thanks{Marcelo~J.\ Weinberger is with Center for the Science of Information,
Stanford University, Stanford, CA 94305 (e-mail: marcwein@ieee.org).}
}
\end{titlepage}
\ifx\proof\undefined%
\newenvironment{proof}{\begin{IEEEproof}}{\end{IEEEproof}}%
\else%
\renewenvironment{proof}{\begin{IEEEproof}}{\end{IEEEproof}}%
\fi%
\newtheorem{theorem}{Theorem}

\newtheorem{fact}{Fact}
\newtheorem{lemma}{Lemma}
\newtheorem{corollary}{Corollary}
\newtheorem{@example}{Example}
\newenvironment{example}{\begin{@example}\upshape}{\hfill$\Box$\end{@example}}

\newtheorem{PROCEDURE}{Procedure}
\newtheorem{REMARK}{\indent
Remark}
\newenvironment{remark}{\begin{REMARK}\rm}{\rm\end{REMARK}}
{\end{enumerate}\end{PROCEDURE}}%
\newcommand{\hruleproc}[1][\columnwidth]{\rule{#1}{1pt}}
\newenvironment{figprocedure}[3][0.9\linewidth]%
{
\begin{center}\begin{minipage}{#1}\small
\hruleproc\\ 
\makebox[3.5em][r]{\textbf{Input:}}
 #2\\
\makebox[3.5em][r]{\textbf{Output:}}
 #3\\ [-1ex]
\hruleproc

\begin{enumerate}
}
{
\end{enumerate}

\vspace{-1ex}
\hruleproc
\end{minipage}\end{center}%
}
\newenvironment{figprocedurenoIO}[1][0.9\linewidth]
{
\begin{center}\begin{minipage}{#1}\small
\hruleproc
\begin{enumerate}%
}
{
\end{enumerate}
\vspace{-1ex}
\hruleproc
\end{minipage}\end{center}%
}

\newcommand{\vartofixlen}{variable to fixed-length}

\newcommand{\fixtovarlen}{fixed to variable-length}

\newcommand{\RNG}{RNG}
\newcommand{\RNGs}{RNGs}
\newcommand{\VFR}{VFR}
\newcommand{\VFRs}{VFRs}
\newcommand{\TVFR}{TVFR}
\newcommand{\TVFRs}{TVFRs}
\newcommand{\FVR}{FVR}
\newcommand{\FVRs}{FVRs}
\newcommand{\dense}{dense}
\newcommand{\truncatedD}{truncated}

\newcommand{\alp}{\mathcal{A}}

\newcommand{\setM}[1][M]{[{#1}]}
\newcommand{\fixsize}{M}
\newcommand{\fixalp}{\setM[\fixsize]}
\newcommand{\dict}{\mathcal{D}}
\newcommand{\idict}{\mathbf{I}_{\mathcal{D}}}

\newcommand{\idictn}{\mathbf{I}_{\mathcal{D}_N
 \cup \failn}}
\newcommand{\fail}{\mathcal{E}}
\newcommand{\chunk}{\mathcal{G}}
\newcommand{\dictn}[1][N]{\dict_{#1}}
\newcommand{\failn}[1][N]{\fail_{#1}}
\newcommand{\xs}{x^\ast}
\newcommand{\Xs}{X^\ast}

\newcommand{\dicts}{\dict^{\ast}}
\newcommand{\fails}{\fail^{\ast}}
\newcommand{\dictns}{\dictn^{\ast}}
\newcommand{\failns}{\failn^{\ast}}

\newcommand{\dictld}{\tilde{\dict}}
\newcommand{\dictlds}{\dictld^\ast}
\newcommand{\dictldi}[1][i]{\Delta_{#1}}
\newcommand{\chunki}[1][i]{\chunk_{#1}}

\newcommand{\astast}{\ast\!\ast}
\newcommand{\fvr}[1][n]{\mathcal{F}_{#1}}
\newcommand{\fvrs}[1][n]{\mathcal{F}^{\ast}_{#1}}
\newcommand{\fvrss}[1][n]{\mathcal{F}^{\astast}_{#1}}
\newcommand{\Mmap}{\mathcal{M}}
\newcommand{\Mmaps}{\Mmap^{\ast}}
\newcommand{\Mmapss}{\Mmap^{\astast}}
\newcommand{\Rmap}{\rho}
\newcommand{\Rmaps}{\Rmap^{\ast}}
\newcommand{\Rmapss}{\Rmap^{\astast}}

\newcommand{\rMinvx}{\chi}
\newcommand{\probM}{\gamma}

\newcommand{\vfr}{\mathcal{V}}
\newcommand{\vfrn}[1][N]{\vfr_{#1}}

\newcommand{\vfrns}[1][N]{\vfr^{\ast}_{#1}}
\newcommand{\vfrs}{\vfr^{\ast}}
\newcommand{\xmap}{\Phi}
\newcommand{\xmapst}{\xmap^{\ast}}
\newcommand{\MM}{\mu}
\newcommand{\notilde}[1]{\bar{#1}}
\newcommand{\tildep}{\bar{p}}

\newcommand{\Pclass}{\mathcal{P}}
\newcommand{\Pclassk}[1][k]{\Pclass_{#1}}
\newcommand{\Pbldp}{P_{\bldp}}
\newcommand{\Exp}{\mathbf{E}}
\newcommand{\tree}{\mathbf{T}}
\newcommand{\treed}{\tree_\dict}

\newcommand{\pDomain}{\Psi}

\newcommand{\len}{L}
\newcommand{\lenp}[1][P]{\len_{#1}}
\newcommand{\lenpn}[1][P]{\lenp[#1]}
\newcommand{\lenpm}[1][P]{\lenp[#1]}

\newcommand{\llen}{\Lambda}

\newcommand{\llenss}{\llen^{\astast}}
\newcommand{\Hrate}{\bar{H}_P(X)}
\newcommand{\HrateB}{\bar{H}}
\newcommand{\Hdict}[1][P]{H_{#1}(\dict)}
\newcommand{\Hdictn}[1][P]{H_{#1}(\dictn)}
\newcommand{\Hdic}[1][P]{H_{#1}}

\newcommand{\bldp}{\mathbf{p}}
\newcommand{\bldq}{\mathbf{q}}

\newcommand{\rnd}{r}
\newcommand{\jT}{j_{\typet}}

\newcommand{\nnints}{\mathbb{N}}
\newcommand{\nonnegints}{\nnints}
\newcommand{\posints}{\mathbb{N}^{+}}
\newcommand{\target}{\mathbb{{N}}_{\textrm{t}}}
\newcommand{\gtarg}[2][\target]{\left\lfloor{#2}\right\rfloor_{#1}}

\newcommand{\typet}{T}
\newcommand{\type}[1][{x^n}]{T({#1})}
\newcommand{\typek}[2][{x^n}]{T^{(#2)}({#1})}
\newcommand{\types}[1][n]{\mathcal{T}_{#1}}
\newcommand{\typesk}[2][n]{\mathcal{T}_{#1}^{(#2)}}

\newcommand{\typext}[2][\typet]{S^{+}_{#2}(#1)}
\newcommand{\typecut}[2][\typet]{S^{-}_{#2}(#1)}
\newcommand{\typecutnot}[2][\typet]{S_{#2}(#1)}

\newcommand{\tcounts}{\mathbf{n}}
\newcommand{\ns}[1][s]{n_{#1}}
\newcommand{\nsa}[1][a]{\ns^{(#1)}}

\newcommand{\pprefix}{\prec}

\newcommand{\remainltr}{\fail}

\newcommand{\indx}{\mathcal{I}}
\newcommand{\medrule}[1][0.9em]{\rule{0pt}{#1}}
\newcommand{\bigbar}[1][0.9em]{\left|\medrule[#1]\right.}
\newcommand{\defined}{\triangleq}
\newcommand{\stacked}[3][1]{\genfrac{}{}{0pt}{#1}{#2}{#3}}

\newcommand{\Pstat}{P^{\rm{st}}}
\newcommand{\Pover}{P_{\rm{\!o/e}}}

\newcommand{\Punder}{P_{\rm{\!u/e}}}
\newcommand{\calU}{{\mathcal
 U}}
\newcommand{\tufvr}{\mathcal{F}^{\rm{\scriptscriptstyle{(TU)}}}_n}
\newcommand{\Mmaptu}{\Mmap^{\rm{\scriptscriptstyle{(TU)}}}}
\newcommand{\Rmaptu}{\Rmap^{\rm{\scriptscriptstyle{(TU)}}}}

\newcommand{\complete}{full}
\newcommand{\completeness}{fullness}
\newcommand{\full}{complete}
\newcommand{\balanced}{balanced}

\newcommand{\biglbracket}[2][1em]{\left#2\rule{0em}{#1}\right.}
\newcommand{\bigrbracket}[2][1em]{\left.\rule{0em}{#1}\right#2}

\begin{document}

\maketitle
\thispagestyle{empty} 

\begin{abstract}
We study random number generators (\RNGs), both in the \fixtovarlen\ (\FVR)
and the \vartofixlen\ (VFR) regimes, in a universal setting in which the
input is a finite memory source of arbitrary order and unknown parameters,
with arbitrary input and output (finite) alphabet sizes. Applying the
method of types, we characterize essentially unique optimal universal
\RNGs\ that maximize the expected output (respectively, minimize the
expected input) length in the \FVR\ (respectively, \VFR) case. For the
\FVR\ case, the \RNG\ studied is a generalization of Elias's scheme, while
in the \VFR\ case the general scheme is new. We precisely characterize, up
to an additive constant, the corresponding expected lengths, which include
second-order terms similar to those encountered in universal data
compression and universal simulation. Furthermore, in the \FVR\ case, we
consider also a ``twice-universal'' setting, in which the Markov order $k$
of the input source is also unknown.
\end{abstract}

\section{Introduction}\label{sec:intro}
Procedures for transforming non-uniform random sources into uniform
(``perfectly random'') ones have been a subject of great interest in
statistics, information theory, and computer science for decades, going back
to at least~\cite{VonNeumann51}. For the purposes of this paper, a
\emph{random number generator} (\RNG) is a deterministic procedure that
takes, as input, samples from a random process over a finite alphabet $\alp$,
and generates, as output, an integer $\rnd$ that is uniformly distributed in
some range $0\le \rnd < M$ (where $M$ may also be a random variable,
depending on the setting, in which case uniformity is conditioned on $M$). If
$M=p^m$, the output $\rnd$ can be regarded as the outcome of $m$ independent
tosses of a \emph{fair $p$-sided coin} (or \emph{die}); when $p=2$, it is
often said, loosely, that the RNG generates $m$ \emph{random bits}.

Various assumptions on the nature of the input process, what is known about
it, and how the samples are accessed, give raise to different settings for
the problem. Regarding the sample access regime, we are interested in two
such settings. In the \emph{\fixtovarlen} \RNG\ (in short, \FVR) setting, it
is assumed that the input consists of a fixed-length prefix $X^n$ of a sample
of the random process, and the size $M$ of the output space depends on $X^n$.
The \emph{efficiency} of the \FVR\ scheme is measured by the expectation,
with respect to the input process, of $\log M$,%
\footnote{Logarithms are to base $2$, unless specified otherwise.} %
the \emph{output length}, which we seek to \emph{maximize}. In other words,
the goal is to obtain as much output ``randomness'' as possible for the fixed
length of input consumed. In the \emph{\vartofixlen} \RNG\ (in short, \VFR)
setting, the size $M$ of the output space is fixed, but the length $n$ of the
input sample is a random variable. The set of input sequences for which such
a scheme produces an output is referred to as a \emph{stopping set} or
\emph{dictionary}. In this case, efficiency is measured by the expectation of
$n$, the \emph{input length}, which we seek to \emph{minimize}. The goal is
to consume as few samples of the input process as possible to produce a
pre-determined amount of ``randomness.''

An \RNG\ is said to be \emph{universal} in a class of processes $\Pclass$ if,
conditioned on the size $M$ of the output range, it produces a uniformly
distributed output for any process in the class (conditioning on $M$ is
trivial in the \VFR\ case). We are interested in universal \RNGs\ that
optimize efficiency (i.e., maximize expected output length in the \FVR\ case,
or minimize expected input length in the \VFR\ case), simultaneously for all
processes $P\in\Pclass$. \FVRs\ have been studied extensively, in various
universal (e.g.~\cite{Elias72,Peres92}) and non-universal
(e.g.,~\cite{VembuVerdu95}) settings.
\VFRs\ were studied in~\cite{VonNeumann51},%
\footnote{The scheme in [1] can also be interpreted as an \FVR, since it
outputs at most one random bit per pair of input symbols. This is the point
of view taken, e.g., in~\cite{Peres92}. In general, \VFRs\ can be used to
construct \FVRs, and vice versa. However, the resulting constructed \RNGs\
will be generally less efficient than if they were optimized for the intended
regime from the start.
} %
\cite{Hoeffding-Simons70,Stout-Warren84} (with emphasis on the case $M=2$ and
Bernoulli sources), and more recently, in more generality, in~\cite{Roche91,
HanHoshi97, Zhou-Bruck-arch}.

Although, in principle, the input length of a \VFR\ is unbounded, we are also
interested in \emph{truncated \VFRs} (\TVFRs). A \TVFR\ either produces a
uniformly distributed output on an input of length $n \le N$, for some fixed
$N$,  or it \emph{fails} (producing no output). We require \VFRs\ to produce
uniform outputs, while admitting some failure probability, \emph{at all
truncation levels} $N$. In that sense, our notion of a universal \VFR\ is
stricter than the one in the earlier literature
(cf.~\cite{VonNeumann51,Hoeffding-Simons70,Stout-Warren84,Roche91,HanHoshi97}),
where generally no conditions are posed on the truncated \VFRs. The stricter
notion may be useful in practical applications, where there is likely to be
some prior knowledge or minimal requirements on the statistics of the source
(e.g., some assurance of minimal ``randomness''). If the \VFR\ has still not
produced an output after consuming an length of input for which this prior
assumption implies that (with high probability) it should have, the whole
system function may be considered suspect (for example, somebody might have
maliciously impaired the random source). With the stricter definition, the
input length threshold can be set arbitrarily, while preserving perfect
uniformity of the \VFR\ output. Although this may seem too fine-grained a
restriction, it will turn out that the penalty it imposes on the expected
input length of the optimal \VFR\ is negligible relative to the main
asymptotic term.

It is well known that in both settings of interest, the entropy rate
$\HrateB$ of the source determines fundamental limits on the asymptotic
performance of the \RNGs, namely, an upper bound on the expected output
length in the \FVR\ case ($n\HrateB$; see, e.g.,~\cite{VembuVerdu95}), and a
lower bound on the expected input length in the \VFR\ case ($(\log
M)/\HrateB$; see, e.g.,~\cite{Roche91} for i.i.d.\ sources). Furthermore, for
various cases of interest, these bounds have been shown to be asymptotically
tight to their main term. We are interested in a more precise
characterization of the performance of optimal algorithms, including the rate
of convergence to the fundamental limits. As in other problems in information
theory, this rate of convergence will depend on the class of processes
$\Pclass$ in which universality is considered. In this paper, we focus on the
class $\Pclassk$ of $k$-th order \emph{finite memory (Markov)} processes,
where the order $k$ is arbitrary.

In the \FVR\ case, we extend the notion of universality, and say that an
\FVR\ is \emph{twice-universal} in the class $\Pclass = \bigcup_k \Pclassk$
of \emph{all} finite memory processes, with both $k$ and the process
parameters unknown, if its output approaches a uniform distribution for every
process in the class (the formal definition of distribution proximity is
provided in Section~\ref{sec:f2v}). We seek twice-universal \FVRs\ that
attain essentially the same efficiency as universal \FVRs\ designed with
knowledge of the order $k$. The relaxation of the uniformity requirement is
necessary to satisfy this strict efficiency goal, as will follow from the
characterization of universal \FVRs\ in Section~\ref{sec:f2v}. To keep the
paper to a reasonable length, we omit the study of twice-universality in the
\VFR\ case, in which the study of the basic universal setting is already
rather intricate.

The contributions of the paper can be summarized as follows. In
Section~\ref{sec:f2v}, we first review results on universal \FVRs.
In~\cite{Elias72}, Elias presented a universal \FVR\ procedure for Bernoulli
processes and binary outputs (i.e., in our notation, the class $\Pclassk[0]$
with $|\alp|=2$ and $M$ a power of $2$). The procedure is optimal in expected
output length, pointwise for every input block size $n$ and every Bernoulli
process. An efficient implementation is described
in~\cite{Ryabko-Matchikina00}, and a generalization to first-order
processes for any $\alp$ is presented in~\cite{Zhou-Bruck12}.%
\footnote{The class studied in~\cite{Zhou-Bruck12} includes the class of
$k$-th order Markov processes, if we interpret a process of order $k$ over
$\alp$ as a process of order $1$ over $\alp^k$. However, the knowledge that
only $|\alp|^{k+1}$ state transition probabilities are nonzero (out of
$|\alp|^{2k}$ in the generic case of this interpretation) can be used with
obvious complexity advantages in the management of transition counts. In
addition, this knowledge is necessary for the precise asymptotics derived in
this paper.
} %
Although not explicitly employing the terminology, these schemes can be seen
as implicitly relying on basic properties of \emph{type
classes}~\cite{Csiszar-Korner,csiszar98method}. We show that Elias's
procedure, when studied explicitly in the more general context of the
\emph{method of types}~\cite{Csiszar-Korner,csiszar98method}, is applicable,
almost verbatim and with a uniform treatment, to broader classes of
processes, and, in particular, to the class $\Pclassk$ for any value of $k$
and any finite alphabet $\alp$, while retaining its universality and
pointwise optimality properties. We precisely characterize, up to an additive
constant, the expected output length of the procedure in the more general
setting. The estimate shows that, similarly to ``model cost'' terms in
universal compression and universal simulation schemes, \FVRs\ exhibit a
second-order term (a term keeping the expected length below the entropy of
the input) of the form $\frac{K}{2}\log n + O(1)$, where $K =
|\alp|^k(|\alp|-1)$ is the number of free statistical parameters in the model
class $\Pclassk$. However, somewhat surprisingly, we observe that this term
is incurred for almost all processes $P$ in the class, even if the \FVR\ is
designed to produce uniform outputs only for $P$. Thus, in the case of \FVRs,
the second-order term is not necessarily interpreted as a ``cost of
universality,'' as is the case in the other mentioned problems. After
reviewing universal \FVRs, in Subsection~\ref{sec:twiceuniversal} we present
a twice-universal \FVR, also inspired on Elias's scheme, but based on the
partition, presented in~\cite{MartinMSW10}, of the space $\alp^n$ into
classes that generalize the notion of type class for the twice-universal
setting. We show that the expected output length of the twice-universal \FVR\
is the same, up to an additive constant, as that of a universal \FVR\
designed for a known order $k$, with the (unnormalized) distance of the
output to a uniform distribution vanishing exponentially fast with the input
length.

We then shift our attention, in Section~\ref{sec:universal VFRs}, to \VFRs\
that are universal in $\Pclassk$, for arbitrary $k$. After formally defining
the setting and objectives, in Subsection~\ref{sec:optimal} we characterize
the essentially unique optimal \VFR, which minimizes both the expected input
length and the failure probability, pointwise for all values of $M$ and at
all truncation levels $N$ (and, thus, also asymptotically). This \VFR\
appears to be, in its most general setting, new, and it can be regarded as an
analogue to Elias's scheme, but for the \vartofixlen\ regime. We precisely
characterize, up to an additive constant, the expected input length of this
optimal \VFR, and show that, as its \FVR\ counterpart, it also exhibits a
second order term akin to a ``model cost.'' In the \VFR\ case, the term is
proportional to $\log\log M$ (the logarithm of the output length), and again
to the number of free statistical parameters of the input model class. We
also show that the failure probability of the optimal \VFR\ decays
exponentially fast with the truncation level $N$. In addition, we show that
the scheme admits an efficient sequential implementation, which is presented
in Subsection~\ref{sec:sequential}. We note that the optimal \VFR\ coincides
with the optimal ``even'' procedure of~\cite{Hoeffding-Simons70} for
Bernoulli processes and $M=2$. A universal \VFR\ that, although sub-optimal
for all $N$, asymptotically attains the entropy limit for Bernoulli processes
and $M=2^m$ was previously described in~\cite{Zhou-Bruck-arch} (without an
analysis of the second order term of interest here). The dictionary
of~\cite{Zhou-Bruck-arch} is a special case of an auxiliary construct we use,
in Subsection~\ref{sec:perbounds}, in the derivation of an upper bound on the
expected length of the optimal VFR.  The optimal scheme itself, however, is
not derived from this construct.

In the computer science literature, \FVRs\ are referred to as
\emph{randomness extractors} (RE), and a deep and extensive theory has
developed, with connections to other fundamental areas of computation (see,
e.g.,~\cite{Shaltiel11} and references therein). We note that the approach we
take to the problem (which follows traditional information-theoretic and
statistical methodology) and the approach taken in the RE literature are
rather different. The results obtained in the RE literature are very powerful
in the sense that they make very few assumptions on the nature of the
imperfect random source, other than a guarantee on the minimum information
content of an input sample, namely, a lower bound on $\min_{x^n}\{-\log
P(x^n)\}$, referred to in the literature as the \emph{min-entropy} of the
source. Results on expected output length are then expressed in terms of this
bound, and are usually tight up to multiplicative constants. By focusing on
specific classes of input sources, on the other hand, we are able to obtain
tighter results (with performance characterized up to additive constants),
and, from a practical point of view, schemes where statistical assumptions
can be reasonably tested (e.g., entropies can be estimated), which does not
appear to be the case with the very broad assumptions employed in RE theory.

\section{Definitions and preliminaries}

\subsection{Basic objects}
\label{sec:strings and trees}

Let $\alp$ be a finite alphabet of size $\alpha = |\alp|$. We denote a
sequence (or string) $x_i x_{i+1}\ldots x_j$ over $\alp$ by $x_i^j$, with
$x_1^j$ also denoted $x^j$, and $x_i^j = \lambda$, the empty string, whenever
$i>j$. As is customary, we denote by $\alp^n$ and $\alp^\ast$, respectively,
the set of strings of length $n$, and the set of all finite strings over
$\alp$, and we let $\xs\in\alp^\ast$ denote a generic string of unspecified
length. For an integer $M> 0$, we let $\setM = \{0,1,\ldots,M-1\}$, and for
integers $t$ and $u$, we write $t \equiv u \pmod{M}$ if $M | (u-t)$, and $t =
u \bmod M$ if $t\equiv u \pmod{M}$ and $0 \le t < M$. We denote by
$\nonnegints$ and $\posints$ the sets of nonnegative and positive integers,
respectively.

An $\alpha$-ary \emph{dictionary} $\dict$ is a (possibly infinite)
prefix-free subset of $\alp^\ast$. We say that $\dict$ is \emph{\complete\/}
if and only if every string $x^n\in\alp^n$ either has a prefix in $\dict$, or
is a prefix of a string in $\dict$.  Naturally associated with a dictionary
$\dict$ is a rooted $\alpha$-ary \emph{tree} $\treed$, whose nodes are in
one-to-one correspondence with prefixes of strings in $\dict$. The leaves of
$\treed$ correspond to the elements of $\dict$, and, for $u\in\alp^\ast$ and
$a\in\alp$ such that $ua$ is a node of $\treed$, $\treed$ has a branch
$(u,ua)$ labeled with the symbol $a$. We identify nodes with their associated
strings, and say that $u$ is the \emph{parent} of $ua$, or, conversely, that
$ua$ is a \emph{child} of $u$. Moreover, we sometimes regard a tree as a set
of sequences and say, e.g., that $\dict$ \emph{is} the set of leaves of
$\treed$. We note that the definition of a \complete\ dictionary is
consistent with the usual definition of a \emph{\complete\ tree} (sometimes
also referred to as a \emph{\full\ tree}). Clearly, $\dict$ is \complete\ if
and only if every node of $\treed$ is either a leaf, or the parent of
$\alpha$ children.\footnote{This condition is generally insufficient for
\completeness\ of infinite trees that do not derive from dictionaries.}
In the sequel, all dictionaries (and corresponding trees) are assumed to be
\complete. Notice that a \complete\ infinite tree need not satisfy the Kraft
inequality with equality~\cite{Linder-Tarokh-Zeger97}. We will observe in
Section~\ref{sec:universal VFRs}, however, that if $\treed$ does not have a
Kraft sum of one, then $\dict$ is rather useless for the construction of an
efficient \VFR. The set of internal nodes of $\treed$ is denoted $\idict$. A
finite tree is \emph{\balanced\/} if it has $\alpha^\ell$ leaves at its
deepest level $\ell$.

\subsection{Finite memory processes and type classes}
A $k$-th order \emph{finite memory (Markov)} process $P$ over the alphabet
$\alp$ is defined by a set of $\alpha^k$ conditional probability mass
functions $p(\,\cdot\,|s):\alp\to[0,1]$, $\,s\in\alpha^k$, where $p(a|s)$
denotes the probability of the process emitting $a$ immediately after having
emitted the $k$-tuple $s$. The latter is referred to as a $\emph{state}$ of
the process, and we assume for simplicity a fixed initial state $s_0$.

Let $b\in\alp$ be a fixed arbitrary symbol. We denote by $\bldp$ the
parameter vector
$\bldp=\left[\,p(a|s)\,\right]_{a\in\alp\setminus\{b\},s\in\alp^k}$, and by
$\pDomain$ its domain of definition. To simplify the statement of some
arguments and results, we further assume that all conditional probabilities
$p(a|s)$ (including $p(b|s)$) are nonzero, i.e., we take $\pDomain$ as an
open
set by excluding its boundary.%
\footnote{
In fact, our results only require that $\pDomain$ be a positive volume subset
of the $\alpha^k(\alpha-1)$-dimensional simplex. The additional requirement
of excluding the boundary guarantees the validity of our asymptotic
expansions for every parameter in $\pDomain$.
} %
The dimension of $\bldp$ is equal to the number of free statistical
parameters in the specification of the $k$-th order process $P$, namely, $K
\defined (\alpha-1)\alpha^k$.

The probability assigned by $P$ to a sequence $x^n = x_1 x_2\dotsc x_n$ over
$\alp$ is
\begin{equation}\label{eq:P}
P(x^n) = \prod_{t=1}^n p(x_t | x_{t-k}^{t-1}),
\end{equation}
where we assume that $x_{-k+1}^0 = s_0$. In cases where we need to consider a
different initial state $s$, we denote the probability $P(x^n|s)$. The
entropy rate of $P \in \Pclassk$ is denoted $\Hrate$ and is given by
\begin{equation}\label{eq:Hrate}
\Hrate = \sum_{s \in \alp^k} \Pstat (s) H_P(X|s)
\end{equation}
where, for a state $s$, $\Pstat(s)$ denotes its stationary probability
(which, by our assumptions on $\Pclassk$, is well defined and nonzero), and
$H_P(X|s)$ denotes its conditional entropy, given by $-\sum_{a \in \alp}
p(a|s) \log{p(a|s)}$.

The \emph{type class} of $x^n$ with respect to the family $\Pclassk$ of all
$k$-th order finite memory processes is defined as the set
\begin{equation}\label{eq:type_prob}
\type = \{\, y^n \in \alp^n \,|\, P(x^n) =
P(y^n)\;\;\forall P \in \Pclassk\,\,\}.
\end{equation}
Let $\nsa(x^n)$ denote the number of occurrences of $a$ following $s$ in
$x^n$, i.e.,
\[
\nsa(x^n) = \left|\{\,t\,\bigbar\,x_{t-k}^{t} = sa,
\; 1\le t \le n\,\}\right|,\;a\in\alp,\,s\in\alp^k,
\]
and $n_s (x^n) \defined \sum_{b \in \alp} \nsa[b](x^n)$. Denote by
$\tcounts(x^n)$ the vector of $\alpha^{k+1}$ integers $\nsa(x^n)$ ordered
according to some fixed convention. It has been well established that the
definition~(\ref{eq:type_prob}) is equivalent to the combinatorial
characterization
\[
\type = \left\{\, y^n \in \alp^n \;\bigbar\;
\tcounts(x^n) = \tcounts(y^n)\,\right\}\,.
\]
The vector $\tcounts(x^n)$ is referred to as the \emph{type} of $x^n$.

The set of all type classes for sequences of length $n$ is denoted $\types$,
i.e.,
\[
\types = \left\{\,\type\;\bigbar\; x^n\in \alp^n\,\right\}\,.
\]

The following fact about type classes $\typet{\,\in\,}\types$ is well known.
\begin{fact}\label{fact:same_sn}
All sequences in $\typet$ share the same final state, i.e., for some fixed
string $u^k\in\alp^k$, we have $x_{n-k+1}^n = u^k$ for all $x^n\in\typet$.
\end{fact}

Unless explicitly stated otherwise, we will assume that the \RNG\
constructions described in this paper have access to the (arbitrary) initial
state $s_0$, and the order $k$ of the processes, which are necessary to
constructively define the type class partitions $\types$. We will depart from
these assumptions in Subsection~\ref{sec:twiceuniversal} when we discuss
twice-universal \RNGs\ (and the order $k$ is not assumed known), and in
Subsections ~\ref{sec:variation Elias} and~\ref{sec:optimal} when we briefly
discuss \RNGs\ that are insensitive to the initial state, and are based on a
slightly different definition of the type class.

Type classes of finite memory processes (and of broader model families) have
been studied extensively  (see, e.g.,~\cite{Csiszar-Korner,csiszar98method}
and references therein). In particular, the cardinality of a type class is
explicitly characterized by \emph{Whittle's formula}~\cite{Whittle}, and a
one-to-one correspondence between the sequences in a type class and Eulerian
cycles in a certain digraph constructed from $\tcounts(x^n)$ was uncovered
in~\cite{Goodman}. Whittle's formula also allows for the computationally
efficient enumeration of the type class, i.e., the computation of the index
of a given sequence in its class, and the derivation of a sequence from its
index, by
means of enumeration methods such as those described in~\cite{Cover73}.%
\footnote{ In this context, ``computationally efficient'' means computable in
polynomial time. Although further complexity optimizations are outside the
scope of this paper, various tools developed for similar enumeration and RNG
problems in the literature would be applicable also here, and should allow
for significant speed-ups.
See, e.g.,~\cite{Ryabko94,Ryabko98,Ryabko-Matchikina00,Zhou-Bruck12}.} %
These enumerations are a key component of the RNG procedures discussed in
this paper.


\section{Universal \fixtovarlen\ RNGs}
\label{sec:Elias}\label{sec:f2v}

\subsection{Formal definition}
\label{sec:definition_fv}

An \FVR\  is formally defined by a triplet $\fvr = (\target,\Rmap,\Mmap)$
where $n \in\posints$ is the fixed input length, $\target\subseteq\posints$
is a \emph{target set} such that $1\in\target$, and $\Rmap:\alp^n\to\nnints$,
$\Mmap:\alp^n\to\target$, are functions such that  $\Rmap(x^n)
\in \setM[\Mmap(x^n)]$. The \emph{output length} of $\fvr$ on input $x^n$ is
defined as $\log\Mmap(x^n)$. Thus, the function $\Mmap$ determines the range
of the output random number and the output length, while the function $\Rmap$
determines the random number itself within the determined range. When the
goal is to generate fair $p$-sided coin tosses, we choose
\begin{equation}\label{eq:coins}
\target =\{\,p^i\,|\,
i\ge 0\,\}\,,\;\quad p \ge 2\,.
\end{equation}
An \FVR\ $\fvr$ is \emph{perfect} for a process $P {\,\in\,}\Pclassk$ if
$\Rmap(x^n)$, conditioned on $\Mmap(x^n)=M$, is uniformly distributed in
$[M]$; $\fvr$ is \emph{universal} in $\Pclassk$ if it is perfect for all
$P\in\Pclassk$. The \emph{expected output length} of $\fvr =
(\target,\Rmap,\Mmap)$ with respect to $P$ is
\begin{equation}
\label{eq:fv_length}
\lenpn(\Mmap){\,\,\defined\,\,}\Exp_P \log \Mmap(X^n){\,\,=} \!\!\sum_{x^n\in\alp^n} P(x^n)\log
 \Mmap(x^n)\,,
\end{equation}
where $\Exp_P$ denotes expectation with respect to $P$. Given a process order
$k$, the goal is to find universal \FVRs\ that maximize $\lenpn$
simultaneously for all $P\in\Pclassk$. We are interested in $\lenpn$ in a
pointwise sense (i.e., for each value of $n$), and also in its asymptotic
behavior as $n\to\infty$.

Notice that our setting is slightly more general than the usual one for
\FVRs\ in the literature, where the condition~(\ref{eq:coins}) for some $p$
is generally assumed in advance. As we shall see, there is not much practical
gain in this generalization. However, the broader setting will allow us to
better highlight the essence of the optimal solutions, as well as connections
to related problems in information theory such as universal compression and
universal simulation.

For conciseness, in the sequel, except when we discuss twice-universality in
Subsection~\ref{sec:twiceuniversal}, when we say ``universal'' we mean
``universal in $\Pclassk$ for a given order $k$, understood from the
context.''

\subsection{Necessary and sufficient condition for universality of \FVRs}
\label{sec:conditions_fv}

The following condition for universality is similar to, albeit stronger than,
conditions previously derived for problems in universal
simulation~\cite{MerhavWeinb04,MSW-delay-limited08,MartinMSW10} and universal
\FVRs~\cite{PaeLoui06,Zhou-Bruck12}. The proof is deferred to
Appendix~\ref{app:condition-fv}.

\begin{lemma}\label{lem:condition-fv} Let $\fvr{=}(\target,\Rmap,\Mmap)$ be
an \FVR\ satisfying the following condition: For all $\typet \in \types$ and
every $M\in\target$, the number of sequences $x^n\in\typet$ such that
$\Mmap(x^n){\,=\,}M$ and $\Rmap(x^n){\,=\,}r$ is the same for all $r\in\setM$
(in particular, the number of sequences $x^n{\,\in\,}\typet$ such that
$\Mmap(x^n)=M$ is a multiple of~$M$). Then, $\fvr$ is universal in
$\Pclassk$. If $\fvr$
\emph{does not} satisfy the condition, then it can only be perfect for
processes $P$ with parameter $\bldp$ in a fixed subset $\pDomain_0$ of
measure zero in $\pDomain$.
\end{lemma}

The following corollary is an immediate consequence of
Lemma~\ref{lem:condition-fv}. It shows that universality is essentially
equivalent to perfection for a single, ``generic'' process in
$\Pclassk$.\footnote{In particular, this result settles a conjecture put
forth in~\cite[p.\ 917]{juels}.}
\begin{corollary}\label{cor:universal=perfect}
An \FVR\ is universal if and only if it is perfect for any single process
$P\in\pDomain\setminus\pDomain_0$, where $\pDomain_0$ is a fixed subset of
measure zero in $\pDomain$.
\end{corollary}

\subsection{Variations on the Elias procedure}
\label{sec:variation Elias}
\newcommand{\EliasI}{E1}
We start by considering the simplest target set, namely $\target = \posints$
(i.e., no restrictions such as~(\ref{eq:coins}) are placed on the ranges of
the generated random numbers). Let $\indx_T(x^n)$ denote the index of
$x^n\in\alp^n$ in an enumeration of~$T=\type$. The following procedure
defines an \FVR\ $\fvrs = (\posints,\Rmaps,\Mmaps)$.

\emph{Procedure \EliasI\/}: Given an input sequence $x^n$, let
$\Mmaps(x^n)=|\type|$, and $\Rmaps(x^n) = \indx_T(x^n)$.

Procedure \EliasI\ is a ``bare-bones'' version of Elias's procedure.%
\footnote{An equivalent procedure is described in~\cite{juels} as a
first step in the implementation of Elias's procedure, the second step
consisting of a ``binarization'' of $\Rmaps(x^n)$, for the case $p=2$
in~(\ref{eq:coins}).} It is straightforward to verify that $\fvrs$ satisfies
the condition of Lemma~\ref{lem:condition-fv} and is, thus, universal in
$\Pclassk$. The following theorem shows that $\fvrs$ attains the maximum
possible expected output length of any universal \FVR\ for the given $n$, all
$P\in\Pclassk$, and arbitrary target set $\target$.

\begin{theorem}\label{thm:exp len simple}
If $\fvr=(\target,\Rmap,\Mmap)$ is universal in $\Pclassk$, then, for any
$P\in\Pclassk$,
\begin{equation}\label{eq:optFVlen}
\lenpn(\Mmap) \le \lenpn(\Mmaps) = \Exp_P\log|\type[X^n]|\,.
\end{equation}
\end{theorem}
\begin{proof}
The equality is straightforward from the definition of Procedure~\EliasI. The
inequality follows from Lemma~\ref{lem:condition-fv} and
Corollary~\ref{cor:universal=perfect}, which imply that $\Mmap(x^n) \le
|\type|$ for all $x^n\in\alp^n$.
\end{proof}

The term on the rightmost side of~(\ref{eq:optFVlen}) was precisely estimated
in~\cite{MerhavWeinb04} in the context of universal simulation of sources in
$\Pclassk$, by analyzing the expectation of Whittle's formula, and obtaining
\begin{equation}\label{eq:Elogtype}
\Exp_P \log |\type[X^n]| = H_P(X^n) - (K/2)\log n + O(1)\,,
\end{equation}
where $H_P(X^n)$ denotes the entropy of the marginal $P(X^n)$,
$K=(\alpha-1)\alpha^k$, and the $O(\cdot)$ notation refers to asymptotics in
$n$.

\begin{remark}
The second-order term $(K/2) \log n$ on the right hand side
of~(\ref{eq:Elogtype}) resembles a typical ``model cost'' term in universal
lossless compression. By Theorem~\ref{thm:exp len simple}, this term
determines the rate at which the expected output length of $\fvrs$ approaches
(from below) $H_P(X^n)$, which is the best possible convergence rate for
\emph{any} universal \FVR. Notice, however, that by
Corollary~\ref{cor:universal=perfect}, the bound of Theorem~\ref{thm:exp len
simple} applies even if the \FVR\ is required to be perfect just for a single
process with parameter $\bldp\in\pDomain \setminus\pDomain_0$. Therefore, in
this case, the second-order term must be incurred (almost always) also in the
non-universal (known $P$) setting, and, in fact, it can be argued that there
is \emph{no asymptotic cost} for universality. Nevertheless, we will still
refer to the second order term as a \emph{model cost}, since it is
proportional to the size of the model, regardless of whether the parameters
of the input process are known or not.
\end{remark}
\begin{remark}
Procedure \EliasI\ is similar to a \emph{universal enumerative encoder}, a
two-part universal lossless compressor for the class $\Pclassk$. The encoder
differs from the \FVR\ in that it outputs, together with $\rho^\ast (x^n)$
and in lieu of $M$, an efficient description of $\type$. It is known (see,
e.g.,~\cite{WMF94}) that $K\log n + O(1)$ bits are sufficient for this
description, resulting in an overall expected code length of $H_P(X^n)
+(K/2)\log n + O(1)$, which is optimal, up to an additive constant, for any
universal lossless compressor for the class $\Pclassk$. The rate of
convergence to the entropy is the same as for \FVRs, but convergence, in this
case, is from above. We observe that, \emph{a fortiori}, a universal lossless
compressor \emph{cannot be} a universal \FVR\ for $\Pclassk$ (and vice
versa).
\end{remark}

We now shift our attention to more general target sets, which include also
sets of the form~(\ref{eq:coins}). Let $\target$ be an arbitrary subset of
the positive integers with $1\in\target$. For any $M \in \posints$, let
\[
\gtarg{M} = \max\left\{\,j\in\target\,|\,j\le M\,\right\}\,.
\]
Let $c$ be a constant, $c \ge 1$. We say that $\target$ is $c$-\emph{\dense}
if and only if for every $M\in\posints$, we have
\[
M \le c\gtarg{M}.
\]
For example, $\posints$ is $1$-\dense\ (no other subset of $\posints$ is),
and the target set in~(\ref{eq:coins}) (used in Elias's procedure for fair
$p$-sided coins) is $p$-\dense. In the sequel, we will assume that $\target$
is $c$-\dense\ for some $c$.

\newcommand{\EliasII}{E2}

Procedure \EliasII\ in \figurename~\ref{figproc:Elias} defines an \FVR\
$\fvrss = (\target,\Rmapss,\Mmapss)$
(we recall that $\indx_T(x^n)$ denotes the index of $x^n$ in an enumeration
of $\type$). The assumption that $1 \in \target$, and the fact that $r < \MM$
holds throughout after the execution of Step~\ref{step:M0}, guarantee that
Procedure \EliasII\ always stops, and, by the stopping condition in
Step~\ref{step:output}, the output satisfies the required condition
$r\in\setM$. It is also readily verified that $\fvrss$ satisfies the
condition of Lemma~\ref{lem:condition-fv} and is, thus, universal.

\begin{figure}
\begin{figprocedure}
{Sequence $x^n\in\alp^n$.}%
{Pair $(\rnd, M),\;\;M\in\target,\;\;\rnd\in\setM$.}%
\item\label{step:M0} Let $\MM = |\type|,\;\;\;\rnd  = \indx_T(x^n)$.
\item\label{step:loop} Repeat forever:
\begin{enumerate}
\item\label{step:M} Let $M = \gtarg{\MM}\;$.
\item\label{step:output} If $\rnd < M $ then
output $(\rnd, M)$ and \textbf{stop.}
\item\label{step:MM-M} Let $\MM = \MM - M,\;\; \rnd = \rnd - M$.
\end{enumerate}
\end{figprocedure}
\caption{\label{figproc:Elias}%
Procedure \EliasII: Generalized Elias procedure ($\fvrss$).}
\end{figure}

We refer to Procedure \EliasII\ as ``greedy,'' since, in
Step~\ref{step:MM-M}, it always chooses to reduce $\MM$ by the largest
possible element of $\target$. The procedure trivially coincides with
Procedure \EliasI\ when $\target=\posints$. When $\target$ is of the
form~(\ref{eq:coins}), the procedure coincides with Elias's original scheme
in~\cite{Elias72}, suitably extended to finite-memory sources of arbitrary
order, and coins with an arbitrary number of sides.

Suppose we do not let Procedure \EliasII\ stop in Step~\ref{step:output},
instead allowing it to run until $\MM = 0$ in Step~\ref{step:MM-M}. Then, the
procedure defines a decomposition of $|\type|$ into a sum
\begin{equation}\label{eq:decomp_mu}
|\type| = \sum_{i=1}^m M_i,
\end{equation}
where $M_i\in\target$, and $M_1 \ge M_2 \ge \cdots \ge M_m$. The term $M_i$
corresponds to the value that $M$ assumes at the $i$-th execution of
Step~\ref{step:M}, namely,
\begin{equation}\label{eq:Mi}
M_i = \gtarg{|\type| - \sum_{j=1}^{i-1}M_{j}},\quad 1 \le i \le m\,,
\end{equation}
where $m$ is the first index such that $M_m \in\target$ ($m$ is well defined
since $1\in\target$).

\begin{remark}\label{rem:minitypes}
Equations~(\ref{eq:decomp_mu})--(\ref{eq:Mi}) define a partition of the
integer $|\type|$ into summands in $\target$, which, through an enumeration
of $\type$, translates to a partition of $\type$ into subclasses, with the
size of each subclass belonging to $\target$. This partition induces a
refinement of the original type-class partition of $\alp^n$, so that all the
sequences in a refined subclass are still equiprobable for all $P \in
\Pclassk$. Procedure~\EliasII\ can then be interpreted as applying
Procedure~\EliasI, but using the refined partition in lieu of the type-class
partition. The procedure stops when it finds the subclass $x^n$ is in, at
which time the value of $r$ is the index of $x^n$ in the subclass.
\end{remark}

Next, we characterize the expected output length of $\fvrss$ when $\target$
is $c$-\dense. The characterization will make use of the following technical
lemma, a proof of which is deferred to Appendix~\ref{app:chcinv}.
\begin{lemma}\label{lem:chcinv}
Let $\bldq=[q_1,\,q_2,\,\ldots,\,q_m]$, with $q_1 \ge q_2 \ge \cdots q_m >
0$, be the vector of probabilities of a discrete distribution on $m$ symbols,
and let $H = -\sum_{i=1}^m q_i \log q_i$ denote its entropy. Assume that for
some constant $c \ge 1$, $\bldq$ satisfies
\begin{equation}\label{eq:qcond}
c\, q_i \ge 1 - \sum_{j=1}^{i-1} q_j,\quad 1 \le i \le m\,.
\end{equation}
Then, letting $h(\cdot)$ denote the binary entropy function, we have $H \le
c\,h(c^{-1})$.
\end{lemma}

\begin{theorem}\label{thm:optFV}
If $\target$ is $c$-\dense, the expected output length of $\fvrss$ for
$P\in\Pclassk$ is
\begin{equation}\label{eq:len c-dense}
\lenpn(\Mmapss) =\lenpn(\Mmaps) - O(1)\,.
\end{equation}
\end{theorem}
\begin{proof}
Let $\typet=\type$ denote an arbitrary type class and let
$M_1,M_2,\ldots,M_m$ denote the integers in $\target$ determined by the
decomposition of $|\typet|$ in~(\ref{eq:decomp_mu})--(\ref{eq:Mi}). Define
\begin{equation}\label{eq:qT}
\bldq(T) = |\typet|^{-1}\,\left[\,M_1,\,M_2,\,\ldots,\,M_m\,\right]\,.
\end{equation}
Clearly, $\bldq(T)$ is the vector of probabilities of a discrete
distribution, with entropy $H(\bldq(T))$.  By~(\ref{eq:Mi}), the $c$-density
assumption applied to the quantities $|\typet| - \sum_{j=1}^{i-1} M_j,\;1 \le
i \le m$, and the definition~(\ref{eq:qT}), Lemma~\ref{lem:chcinv} applies to
$\bldq(T)$. Since the sequences in a type class are equiprobable, the
expectation of $\log \Mmapss(X^n)$ conditioned on $T$ is given by
\begin{equation}\label{eq:L(T)}
 \len(T) \defined |T|^{-1}\sum_{i=1}^m M_i \log M_i
= \log|T| - H(\bldq(T)) \,,
\end{equation}
and
\begin{align}
\lenpn(\Mmapss) & = \sum_{T\in\types}P(T)\len(T)
\ge E_P\log|T| - c\,h(c^{-1})\nonumber\\
& = \lenpn(\Mmaps) - O(1)\,,
\label{eq:arriba}
\end{align}
where the inequality follows from~(\ref{eq:L(T)}) and Lemma~\ref{lem:chcinv},
and the last equality follows from the rightmost equality
in~(\ref{eq:optFVlen}). The claim of the theorem now follows by combining the
lower bound~(\ref{eq:arriba}) with the upper bound in Theorem~\ref{thm:exp
len simple}.
\end{proof}

\iftrue  
\begin{remark}\label{rmk:c-dense}
In a worst-case sense, the sufficient condition of $c$-density in the theorem
is also necessary, since, using the fact that $\len(T)\le \log M_1$
by~(\ref{eq:L(T)}) (with the notation in the proof), we have
\begin{equation}\label{eq:unbounded}
\log|T| - \len(T) \ge  \log|T| - \log M_1 = \log\frac{|T|}{\gtarg{|T|}}\,.
\end{equation}
Thus, if $\target$ is not $c$-\dense\ for any $c$ then the expression on the
rightmost side of~(\ref{eq:unbounded}) is unbounded.
\end{remark}
\fi

Theorem~\ref{thm:optFV} shows that if $\target$ is $c$-\dense,  $\fvrss$
performs to within an additive constant of the expected output length of
$\fvrs$, which is an upper bound for any universal \FVR, independently of the
target set. In particular, this implies that $\fvrss$ is optimal, up to an
additive constant, among all \FVRs\ for the same target set $\target$. While,
for a general $c$-dense target set, this additive constant is positive, the
following theorem shows that when $\target$ is of the form~(\ref{eq:coins}),
$\fvrss$ is in fact the optimal \FVR\ for $\target$. This result was proved
for $k{=}0$ in~\cite{PaeLoui06},~\cite{juels}, and for $k{=}1$
in~\cite{Zhou-Bruck12}. In fact, once the basic properties of type classes
are established, the proof should be rather insensitive to the order $k$, as
it follows, essentially, from Lemma~\ref{lem:condition-fv},
from~(\ref{eq:type_prob}), and from the fact that for an arbitrary positive
integer $\MM$, the sum $\sum_{i=0}^m i\, c_i p^i$, subject to $\sum_{i=0}^m
c_i p^i = \MM < p^{m+1}$, is maximized over vectors of nonnegative integers
$[c_0,c_1,...,c_m]$ when $c_0,c_1,...,c_m$ are the digits in the radix-$p$
representation of $\MM$ (in our case, $\MM$ corresponds to the size of a type
class).

\begin{theorem}\label{thm:Elias optimal}
Let $\target = \{ p^i \,|\, i\ge 0\,\}$ for some integer $p\ge 2$. Consider
$\fvrss$ with target set $\target$, and let $\fvr = (\target,\Rmap,\Mmap)$ be
any \FVR\ with the same target set. Then, for any $n$ and any $P\in\Pclassk$,
we have
\[
\lenpn(\Mmap) \le \lenpn(\Mmapss)\,.
\]
\end{theorem}

\begin{remark}\label{rmk:initialFVR}
The proposed variants of the Elias procedure assume knowledge of the
(arbitrary) initial state $s_0$. If the initial state is unknown (possibly
nondeterministic), the procedure can consume $k$ input symbols to synchronize
its state and start at $x_{k+1}$, thus generating $\Exp_P \log
|\type[X_{k+1}^n]|$ random bits (up to an additive constant), which is still
asymptotically optimal. However, the pointwise optimality of
Theorem~\ref{thm:Elias optimal} is lost. Nevertheless, the modified procedure
can still be shown to be pointwise optimal in a more restrictive sense for
the target sets covered by the theorem. Specifically, consider a setting in
which an \FVR\ is said to be universal if it is perfect for every $P \in
\Pclassk$ and every initial state distribution (equivalently, for every fixed
initial state). The definition of a type class in~(\ref{eq:type_prob}) is
modified accordingly, and it is easy to see that the corresponding
combinatorial characterization is given by the set of sequences with a given
$\tcounts(x^n)$ and fixed $x_1^k$. It can also be shown that, with the
addition of $\alpha^k-1$ free parameters (those corresponding to the
distribution on the initial state), the type probabilities remain linearly
independent, as required by the proof of Lemma~\ref{lem:condition-fv}. It
follows that the modified \FVR\ is optimal among \FVRs\ that are perfect for
every $P \in \Pclassk$ and every initial state distribution.
\end{remark}

\subsection{Twice-universal \FVRs}
\label{sec:twiceuniversal} In this subsection, we assume that the order $k$
of the Markov source is not known, yet we want to produce a universal \FVR\
whose model cost is not larger (up to third order terms) than the one we
would incur had the value of $k$ been given. To this end, as mentioned in
Section~\ref{sec:intro}, we need to relax our requirement of a uniformly
distributed output. This is necessary since, by Theorem~\ref{thm:exp len
simple} and~(\ref{eq:Elogtype}), an \FVR\ that is universal in $\Pclassk$
would incur a model cost of the form $(K/2)\log n$, with $K = K(k) =
\alpha^k(\alpha-1)$, for any process in the class, including those of orders
$k'<k$, which form a subclass of $\Pclassk$. However, for such a subclass, we
aspire to achieve a smaller model cost proportional to $K(k')$.\footnote{Of
course, application of Procedure~\EliasII\ with $k$ replaced with a slowly
growing function of $n$ leads, for $n$ sufficiently large, to a perfect \FVR\
for any (fixed, but arbitrary) Markov order. However, the model cost incurred
does not meet our efficiency demands.} We assume throughout that $\target$ is
$c$-dense and that the fixed string determining the initial state is as long
as needed (e.g., a semi-infinite all-zero string).

Let  $Q_M (r)$ denote the output probability of $r{\in}\setM$, $M\in\target$,
conditioned on $\Mmap(x^n){=}M$, for an \FVR\ $\fvr=(\target,\Rmap,\Mmap)$.
The distance of $\fvr$ to uniformity is measured by
\begin{equation} \label{eq:distortion}
D(\fvr) \defined \sum_{M \in \target} \frac{P(\Mmap(X^n)=M)}{M}
\!\!\! \sum_{r,r' \in \setM} \!\! \left | Q_M(r) - Q_M(r') \right | \,.
\end{equation}
For any distribution $R(\cdot)$ with support $\mathcal{B}$, we have
\begin{align*}
\sum_{x \in \mathcal{B}} \left| R(x) - \frac{1}{|\mathcal{B}|} \right|
&= \frac{1}{|\mathcal{B}|}
\sum_{x\in\mathcal{B}}\bigg|\sum_{y\in\mathcal{B}}\big(R(x)-R(y)\big)\bigg|\\
&\le \frac{1}{|\mathcal{B}|} \sum_{x,y \in \mathcal{B}} \bigg| R(x) - R(y) \bigg| \,.
\end{align*}
In particular, the inner summation in~(\ref{eq:distortion}) is lower-bounded
by $M \sum_{r \in \setM} | Q_M(r) - 1/M |$. Therefore, our measure of
uniformity is more demanding than the weighted $L_1$ measure used
in~\cite{VembuVerdu95}. Notice that, as in~\cite{VembuVerdu95}, the
measure~(\ref{eq:distortion}) is \emph{unnormalized}. We aim at \FVRs\ for
which $D(\fvr)$ vanishes exponentially fast with $n$.

As in~\cite{MartinMSW10}, our twice-universal \FVR\ will rely on the
existence of Markov order estimators with certain consistency properties,
which are specified in Lemma~\ref{lemma:est} below. For concreteness, we will
focus on a penalized maximum-likelihood estimator that, given a sample $x^n$
from the source, chooses order $k(x^n)$ such that
\begin{equation} \label{eq:bic}
k(x^n) = \arg\min_{k \ge 0} \left\{\, \hat{H}_k (x^n) + \alpha^k \varphi(n) \,\right\}
\end{equation}
where $\hat{H}_k (x^n)$ denotes the $k$-th order empirical conditional
entropy for $x^n$, $\varphi(n)$ is a vanishing function of $n$, and ties are
resolved, e.g., in favor of smaller orders. For example, $\varphi(n)
\,{=}\,(\alpha\,{-}\,1)(\log n)/(2n)$ corresponds to the asymptotic version
of the MDL criterion~\cite{Riss84}. The estimate $k(x^n)$ can be obtained in
time that is linear in $n$ by use of suffix trees~\cite{BB03,MartSerWein04}.
The set of $n$-tuples $x^n$ such that $k(x^n) \,{=}\,i$ will be denoted
$\alp_i^n$. To state Lemma~\ref{lemma:est} we define, for a distribution $P
\,{\in}\, \Pclassk$, the overestimation probability
\[
\Pover (n) \defined P(k(X^n) > k)
\]
and, similarly, the underestimation probability
\[
\Punder (n) \defined P(k(X^n) < k) \,.
\]
Since the discussions in this subsection involve type classes of varying
order, we will use the notation $\typek{k}$ to denote the type class of $x^n$
with respect to $\Pclassk$.
\begin{lemma}[\cite{MartinMSW10}] \label{lemma:est}
For any $k \ge 0$ and any $P \,{\in}\, \Pclassk$, the estimator
of~(\ref{eq:bic}) satisfies
\begin{itemize}
\item [(a)] $(n+1)^{\alpha^{k+1}} \Pover(n)$ vanishes polynomially fast
    (uniformly in $P$ and $k$) provided $\varphi(n) \,{>}\, \beta (\log
    n)/n$ for a sufficiently large constant $\beta$.
\item [(b)] $\Punder(n)$ vanishes exponentially fast.
\end{itemize}
\end{lemma}

Following~\cite{MartinMSW10}, we consider a partition of $\alp^n$ in which
the class of $x^n$, denoted $U(x^n)$, is given by
\begin{equation} \label{eq:tuclass}
U(x^n) \defined \typek{k(x^n)} \cap \alp_{k(x^n)}^n \,.
\end{equation}
Thus, two sequences are in the same class if and only if they estimate the
same Markov order and are in the same type class with respect to the
estimated order. Our twice-universal \FVR, $\tufvr =
(\target,\Rmaptu,\Mmaptu)$, is given by replacing, in Procedure \EliasII,
$\typek{k}$ with $U(x^n)$ and $\indx_T(x^n)$ with the index of $x^n$ in an
enumeration of $U(x^n)$.

\begin{theorem}\label{thm:TUElias}
For $P\in\Pclassk$, the \FVR\ $\tufvr$ satisfies $D(\tufvr)\le 2\Punder(n)$,
and, for a suitable choice of $\varphi(n)$ in~(\ref{eq:bic}), its expected
output length $\lenpn(\Mmaptu)$ satisfies
\begin{equation} \label{eq:tuout}
\lenpn(\Mmaptu)-\lenpn(\Mmaps) = O(1)
\end{equation}
provided $\target$ is $c$-\dense.
\end{theorem}

\begin{remark}
By Lemma~\ref{lemma:est}, Theorem~\ref{thm:TUElias} states that the distance
of $\tufvr$ to uniformity is exponentially small whereas,
by~(\ref{eq:optFVlen}) and~(\ref{eq:Elogtype}), its expected output length is
essentially the same as that of $\fvrs$. It should be pointed out, however,
that Theorem~\ref{thm:TUElias} falls short of stating that the cost of
twice-universality in terms of expected output length is asymptotically
negligible. The reason is that, in principle, it could be the case that by
allowing a small deviation from uniformity, as we do, we open the door for
schemes that (with knowledge of $k$) produce an output significantly longer
than $\fvrs$. We conjecture that, just as in twice-universal
simulation~\cite{MartinMSW10}, this is not the case.
\end{remark}

\begin{remark}
One problem in the implementation of $\tufvr$ is that it requires a
computationally efficient enumeration of $U(x^n)$. Such an enumeration
appears to be elusive. Instead, the following \FVR\ can be efficiently
implemented: Compute $k(x^n)$ and apply Procedure~\EliasII\ with
$k{=}k(x^n)$. A variant of the proof of Theorem~\ref{thm:TUElias} shows that
the output length of this scheme still satisfies~(\ref{eq:tuout}), whereas
its distance to uniformity is upper-bounded by $4 [\Punder (n) {+} \Pover
(n)]$. By Lemma~\ref{lemma:est}, this means that a suitable choice of
$\varphi(n)$ still guarantees vanishing distance, but we can no longer claim
it to be exponentially small.
\end{remark}

\begin{proof}[Proof of Theorem~\ref{thm:TUElias}]
Let $\calU$ denote the set of classes in the refinement of the
partition~(\ref{eq:tuclass}) determined by Procedure \EliasII\ (see
Remark~\ref{rem:minitypes}), and let $\calU_M$ denote the subset of $\calU$
formed by classes of size $M \in \target$. For $U \in \calU_M$, let
$\rho_U^{-1} (r)$ denote the unique sequence in $U$ such that
$\Rmaptu(\rho_U^{-1} (r)) = r$, $r \in \setM$. Let $Q(r,M)$ denote the
probability that $\Mmaptu(x^n)=M$ and $\Rmaptu(x^n)=r$, $M \in \target$, $r
\in \setM$, so that
\[
Q_M(r) = \frac{Q(r,M)}{\sum_{j \in \setM} Q(j,M)} =
\frac{Q(r,M)}{P(\Mmap(X^n)=M)}\,.
\]
 Clearly,
\[
Q(r,M) = \sum_{U \in \calU_M} P(\rho_U^{-1} (r)) \,.
\]
By~(\ref{eq:distortion}),
\begin{eqnarray*}
D(\tufvr) & \!\!\!\!\! =& \!\!\!\!\! \sum_{M \in \target} \! \frac{1}{M} \!
\sum_{r,r' \in \setM} | Q(r,M) - Q(r',M) | \\
&\!\!\!\!\!\le& \!\!\!\!\! \sum_{M \in \target} \! \frac{1}{M} \! \sum_{U \in \calU_M}
\sum_{r,r' \in \setM} | P(\rho_U^{-1} (r)) \!-\! P(\rho_U^{-1} (r')) |
\end{eqnarray*}
which, given the existence of a one-to-one correspondence between $U \in
\calU_M$ and $\setM$, takes the form
\begin{equation} \label{eq:multsum}
D(\tufvr) \le \sum_{M \in \target} \frac{1}{M} \sum_{U \in \calU_M} \sum_{u,v \in U}
| P(u) - P(v) | \,.
\end{equation}
Now, since $U$ is a subset of a type class $T \in \typesk{k(x^n)}$, we have
$P(u)=P(v)$ for all $u,v \in U$ whenever $k(x^n) \ge k$. In addition, we have
the following lemma, which is proved in Appendix~\ref{app:aux}.
\begin{lemma} \label{lem:aux}
For any distribution $R(\cdot)$ on a set containing $\mathcal{B}$, we have
\[
\sum_{u,v \in \mathcal{B}} | R(u) - R(v) | \le 2 (|\mathcal{B}|-1) R(\mathcal{B})\,.
\]
\end{lemma}
Therefore, letting $\calU_M^{\rm{u/e}}$ denote the subset of $\calU_M$ formed
by all the classes such that $k(x^n) < k$, (\ref{eq:multsum}) implies
\[
D(\tufvr) \le \sum_{M \in \target} \frac{2(M-1)}{M} \sum_{U \in \calU_M^{\rm{u/e}}}
P(U) \le 2 \Punder (n) \,,
\]
as claimed.

To lower-bound the expected output length, we first discard the output length
produced by sequences which are not in $\alp_k^n$. We then note that the
claim of Theorem~\ref{thm:optFV} is valid not only for expectations
conditioned on a type (as implicit in its proof, see~(\ref{eq:L(T)})), but
also when conditioning on subsets of types, thus obtaining
\[
\lenpn(\Mmaptu) \ge \sum_{T \in \typesk{k}} P(T \cap \alp_k^n)
\log |T \cap \alp_k^n| + O(1) \,.
\]
By~\cite[Lemma 1]{MartinMSW10}, the number of sequences in a type class $T$
that estimate order $k$ is $|T|-o(1)$ for suitable choices of $\varphi(n)$,
provided that at least one sequence in $T$ estimates order $k$ (i.e., almost
all the sequences in the type class estimate the right order). Therefore,
\begin{align*}
\lenpn(\Mmaptu) \ge  \; \Exp_P &\log |T_k(X^n)|\\
&  - n \left[\Punder (n) + \Pover (n)\right]
\log \alpha + O(1) \,,
\end{align*}
where we have also used the trivial bound $\log |T(x^n)| \le n \log \alpha$
for sequences $x^n$ outside $\alp_k^n$. The claim then follows from
Lemma~\ref{lemma:est}.
\end{proof}


\section{Universal \vartofixlen\ RNGs}
\label{sec:universal VFRs}

\subsection{Formal definition and preliminaries}
\label{sec:definition_vf}

We recall from Subsection~\ref{sec:strings and trees} that a dictionary is a
(possibly infinite) prefix-free set of finite strings $\dict
\subseteq\alp^\ast$, which we assume \complete. A \VFR\ is formally defined
by a triplet $\vfr = (\dict,\xmap,M)$ where $\dict$ is a dictionary, $M > 1$
is a fixed integer, and $\xmap$ is a function $\xmap:\dict\to\fixalp$. For
$N\ge 1$, the restriction to level $N$ of $\dict$ is
\[
\dictn = \{\,x^n\in\dict\,|\,n\le N\,\}\,.
\]
Associated with $\dictn$ is a \emph{failure set} $\failn$, defined as
\[
\failn = \left\{\,x^N\in\alp^N\;\bigbar\;
x^N \;\text{has no prefix in }\;\dictn\,\right\}\,.
\]
The strings in $\dictn\cup\failn$ are identified with the leaves of a finite
\complete\ tree, which is the truncation to depth $N$ of the tree, $\treed$,
corresponding to $\dict$. Nevertheless, we will slightly abuse terminology,
and refer to $\dictn$ (alone) as a \emph{truncated dictionary}. Notice that
$\bigcup_{N \ge 1} \failn$ corresponds to the set of internal nodes $\idict$
of $\treed$, whereas we recall that $\dict$ corresponds to the set of leaves
of $\treed$.

The \VFR\ $\vfr$ generates random numbers from a process $P$ by reading
symbols from a realization of $P$ until a string $x^n$ in $\dict$ is reached,
at which point $\vfr$ outputs $\xmap(x^n)$. The \emph{truncated
\VFR} (\TVFR) $\vfrn = (\dictn, \xmap, M)$, operates similarly, except that
it restricts the length of the input string to $n\le N$, so that $\xmap$ is
applied only to strings in $\dictn$, and the input may reach strings
$x^N\in\failn$, in which case the \TVFR\ \emph{fails} and outputs nothing.

A \VFR\ $\vfr = (\dict,\xmap,M)$ is \emph{perfect} for $P\in\Pclassk$ if for
every $n\ge 1$, either $\dictn[n]$ is empty, or $\xmap(X^n)$, conditioned on
$X^n\in\dictn[n]$, is uniformly distributed in $\fixalp$; $\vfr$ is
\emph{universal} in $\Pclassk$ if it is perfect for all $P\in\Pclassk$. By
extension, we also refer to a \TVFR\ that satisfies the same properties up to
a certain length $N$ as perfect or universal, as appropriate.

Next, we introduce tools that are instrumental in setting our objective. Let
the dictionary $\dict$ satisfy $\sum_{\xs\in\dict} P(\xs) = 1$ for all $P
\in\Pclassk$. Notice that, if $\dict$ is finite, this condition is trivially
satisfied by \completeness; however, as discussed in
Subsection~\ref{sec:strings and trees}, it may not hold for a \complete\
infinite dictionary (for which the summation is understood as an infinite
series in the usual manner), as it was shown in~\cite{Linder-Tarokh-Zeger97}
that the Kraft inequality may be strict. Notice also that the condition is
equivalent to $P(\failn) \stackrel{N \to \infty} {\longrightarrow} 0$. Let
$f$ be a real-valued function of $\xs \in \alp^\ast$. The expectation of $f$
over $\dict$ is denoted $\Exp_{P,\dict}\;f(X^\ast)$ and, in case $\dict$ is
infinite, it is given by
\begin{equation}\label{eq:expf}
\Exp_{P,\dict} \; f(X^\ast) =
\lim_{N \to \infty} \sum_{\xs \in \dictn} P(\xs) f(\xs) \,.
\end{equation}
If $f$ satisfies $0 \le f(y^\ast) \le f(\xs)$ for every prefix $y^\ast$ of
$\xs$ (which is the case for functions such as string length or
self-information), then it is easy to see that, due to the \completeness\ of
$\dict$ and to the vanishing failure probability, we have
\begin{equation}\label{eq:expunion}
\Exp_{P,\dict} \; f(X^\ast) \ge \Exp_{P,\dictn\cup\failn} f(X^\ast)
\end{equation}
for every $N{>}0$, provided the sequence on the right-hand side
of~(\ref{eq:expf}) converges. But, since $\dictn\subseteq\dictn\cup\failn$,
the reverse inequality must hold when we let $N\to\infty$ on the right-hand
side of~(\ref{eq:expunion}). Therefore, the expectation also takes the form
\begin{equation}\label{eq:exptruncf}
\Exp_{P,\dict} \; f(X^\ast) =
\lim_{N \to \infty} \Exp_{P,\dictn\cup\failn} f(X^\ast) \,.
\end{equation}

A useful tool in the analysis of $\Exp_{P,\dict}\;f(X^\ast)$ is the so-called
\emph{leaf-average node-sum interchange theorem} (LANSIT)~\cite[Theorem
1]{Rueppel94}, which states that
\begin{align}
\Exp_{P,\dict}\;& f(X^\ast) \nonumber \\
 &= \sum_{\xs \in \idict} P(\xs) \sum_{a \in \alp}
P(a|\xs) [f(\xs a)-f(\xs)] - f(\lambda) \,.
\label{eq:lansit}
\end{align}
In particular, for $f(\xs)=|\xs|$, the LANSIT reduces to the well-known fact
that the expected depth of the leaves of a tree equals the sum of the
probabilities of its internal nodes. We will use the LANSIT also for
$f(\xs)=1/\log P(\xs)$ and $f(\xs)= n_s (\xs)$, $s \in \alp^k$. The proof of
the theorem, by induction on the number of nodes, is straightforward.

Consider a \VFR\ $\vfr$ that is perfect for $P$. The quantity
\begin{align}
\lenpm(\dictn) &\defined \Exp_{P,\dictn\cup\failn} \; |X^\ast|\nonumber\\
& =
\sum_{n=1}^N \sum_{x^n \in \dictn} n P(x^n) + N P(\failn)
\label{eq:explen}
\end{align}
is an appropriate figure of merit for $\vfr$ at truncation level $N$, as it
measures the \emph{expected input length}, namely the amount of ``raw''
random data that the \VFR\ consumes in order to produce a perfectly uniform
distribution on $\fixalp$, when restricted to inputs of length at most $N$.
The expected input length includes the cost of ``unproductive'' input that is
consumed when the truncated \VFR\ fails (second term on the rightmost side
of~(\ref{eq:explen})). The figure of merit for $\vfr$ is given by
\begin{equation}\label{eq:len_vfr}
\lenpm(\dict) = \lim_{N \to \infty} \lenpm(\dictn)
\end{equation}
which, by~(\ref{eq:exptruncf}), coincides with the expected dictionary length
provided $\sum_{\xs\in\dict} P(\xs) = 1$.\footnote{If, instead, $\sum_{\xs
\in \dict} P(\xs) < 1$, it may be the case that the expected dictionary
length converges (again,~\cite{Linder-Tarokh-Zeger97} provides an example of
such a tree) while, clearly, the expected input length diverges. In this
case, the expected dictionary length is of no interest since, with a positive
probability, an input sample will not have a prefix in the dictionary (i.e.,
the \VFR\ will not stop). }

We are interested in universal \VFRs\ that minimize these measures
simultaneously for all $P\in\Pclassk$, either in a pointwise sense, i.e.,
minimizing $\lenpm(\dictn)$ for all $N$, or asymptotically, i.e., minimizing
the limit $\lenpm(\dict)$. A secondary objective is to minimize the failure
probability $P(\failn)$. Finally, we are interested in \emph{computationally
efficient implementations}, namely, \VFR\ procedures that process the input
sequentially, and run in time polynomial in the consumed input length.

By the LANSIT, we have
\begin{equation}\label{eq:internal}
\lenpm(\dictn) = \sum_{\xs \in \idictn} P(\xs) = \sum_{n=0}^{N-1} P(\failn[n])\,.
\end{equation}
Therefore, the limit in~(\ref{eq:len_vfr}) exists if and only if
$P(\failn[n])$ is summable. We will show that, in fact, the failure
probability in our constructions vanishes exponentially fast, so that the
limit does exist (and equals the expected dictionary length).

In the sequel, we will make extensive use of the following notation: For
$\typet \in \types$ and $\mathcal{S} \subset \alp^\ast$, $\mathcal{S}(\typet)
\defined \mathcal{S} \cap \typet$ (this definition is extended to the case in
which $\mathcal{S}$ is a set of nodes in a tree). In particular, we have
$\idict(T) = \failn[n](\typet)$ and $\treed (T) =
\dictn[n] (T) \cup \failn[n] (T)$.

\subsection{Necessary and sufficient condition for universality of \VFRs}
The analog of Lemma~\ref{lem:condition-fv} for \VFRs\ is given in the
following lemma. The proof is similar, and is presented, for completeness, in
Appendix~\ref{app:condition-vf}.

\begin{lemma}\label{lem:condition-VF}
Let $\vfr = (\dict,\xmap,M)$ be a \VFR\ satisfying the following condition:
For every $n$ and every $\typet\in\types$, the number of sequences
$x^n\in\dict(T)$ such that $\xmap(x^n) = r$ is the same for all $r \in
\fixalp$ (in particular, $|\dict(T)|$ is a multiple of $M$). Then, $\vfr$ is
universal in $\Pclassk$. If $\vfr$ \emph{does not} satisfy the condition,
then it can only be perfect for processes $P$ with parameter $\bldp$ in a
subset $\pDomain'_0$ of measure zero in $\pDomain$.
\end{lemma}

An analog of Corollary~\ref{cor:universal=perfect} for the \VFR\ case is also
straightforward. Notice that if $|\dict(T)|$ is a multiple of $M$, then it is
trivial to define $\xmap$ so that $\vfr$ satisfies the condition of the
lemma. Therefore, designing a universal \VFR\ is essentially equivalent to
designing a dictionary $\dict$ such that
\begin{equation}\label{eq:condition-VF}
|\dict(\typet)| = \jT M\,,\;\;\forall T\in\types,\;\forall n\in\posints\,,
\end{equation}
where $\jT$ is a nonnegative integer dependent on $\typet$. In fact, in our
discussions, we will focus on the condition~(\ref{eq:condition-VF}) and
assume that a suitable mapping $\xmap$ is defined when the condition holds.

\begin{remark}\label{rem:even}
Lemma~\ref{lem:condition-VF} implies that our universal \VFRs\ are akin to
the \emph{even procedures} discussed in~\cite{Hoeffding-Simons70}
and~\cite{Stout-Warren84} (the term \emph{even} derives from the fact that
the emphasis in~\cite{Hoeffding-Simons70} is on the case $M=2$, although the
more general case is also mentioned). In our case, the necessity of the
condition~(\ref{eq:condition-VF}) stems from our requirement that the \VFR\
be perfect at every truncation level $N$. When this requirement is relaxed,
the condition need no longer hold, as evidenced by some of the procedures
presented in~\cite{Hoeffding-Simons70} and \cite{Stout-Warren84}. As we will
see, such a relaxation may reduce the expected input length of the optimal
\VFR\ only by a negligible amount relative to the main asymptotic term
(see also Example~\ref{ex:ex1} below).
\end{remark}

\begin{remark}\label{rem:freedom}
Notice that the condition on universality in Lemma~\ref{lem:condition-VF}
depends only on the \emph{sizes} of the sets $\dict(\typet)$, but not on
their composition. Clearly, the same holds for the expected length and the
failure probability of a (\truncatedD) dictionary, since sequences of the
same type have the same length and probability.  We conclude that the main
properties of interest for a \VFR\ are fully determined by the \emph{type
profile} of its dictionary, namely, the sequence of numbers
$\left\{\,|\dict(\typet)|\,\right\}_{\typet\in\types,\;n\ge 1}$.
\end{remark}

Define
\begin{equation}\label{eq:N0}
N_0(M) = \min\left\{\,n\,\bigbar\,\exists\,
\typet{\,\in\,}\types \text{ such that } |\typet|\ge M\,\right\}.
\end{equation}

An immediate consequence of Lemma~\ref{lem:condition-VF} is that if $\dict$
is the dictionary of a universal \VFR, then $n\ge N_0(M) \ge (\log
M)/(\log\alpha)$ for every $x^n\in\dict$, where the second inequality follows
from~(\ref{eq:N0}) and $|T|\le\alpha^n$.

\subsection{Optimality of a ``greedy'' universal \VFR}
\label{sec:optimal}

We describe the (conceptual) construction of a universal \VFR, and prove its
optimality. The construction is ``greedy'' in the sense that, at every point,
it tries to add to the dictionary as many sequences of a given length as
allowed by the necessary condition of Lemma~\ref{lem:condition-VF}, and by
the prefix condition. In this sense, the procedure can be seen as a
counterpart, for \VFRs, to Elias's scheme for \FVRs\ (recall the discussion
on the ``greediness'' of Procedure \EliasII\ in Subsection~\ref{sec:variation
Elias}). As in the \FVR\ case, it will turn out that greediness pays off, and
the constructed \VFR\ will be shown to be optimal in a pointwise,
non-asymptotic sense. The difficulty in establishing this optimality will
reside in the fact that sequences that get included in $\dict$ ``block'' all
of their continuations from membership in $\dict$. It seems possible, in
principle, that it might pay off not to include some sequences of a given
length, even though the conditions governing the construction allowed their
inclusion, so as to increase our choices for longer sequences. We will prove
that, in fact, this is not the case.

\newcommand{\GreedyI}{G1}
\begin{figure}
\begin{figprocedure}
{Integers $M\ge2,\;N\ge 1$.}%
{\TVFR\ $\vfrns=(\dictns,\xmapst,M)$.}

\item Set $n = 1$, $\dictns = \emptyset$, $\failns = \alp$.
\item\label{step:loop_types} For each type class $T\in\types$, do:
\begin{enumerate}
\item\label{step:divide} Let $\jT =
    \left\lfloor\medrule|\failns(T)|/\fixsize\right\rfloor$. Select
    any subset of $\jT\fixsize$ sequences from $\failns(T)$, add them to
    $\dictns$, and remove them from $\failns(T)$.

\item\label{step:map} Let $\indx(y^n)$ denote the index of $y^n
    \in\dictns(T)$ in some ordering of $\dictns(T)$. Define
\[
\xmapst(y^n) = \indx(y^n)\,\bmod\, \fixsize\,,\quad y^n\in\dictns(T).
\]
\end{enumerate}
\item \label{step:stop} If $n = N$, \textbf{stop}. Otherwise, for
    each sequence $x^n \in \failns$, remove $x^n$ and add all the
    sequences in $\{ x^n a,\, a \in \alp \}$ to $\failns$. Set
    $n \leftarrow n+1$ and go to Step~\ref{step:loop_types}.
\end{figprocedure}
\caption{\label{figproc:greedy}Procedure \GreedyI: Greedy \TVFR\
construction.}
\end{figure}

Procedure \GreedyI\ in \figurename~\ref{figproc:greedy} shows the
construction of a \TVFR\ $\vfrns=(\dictns,\xmap,M)$. The \VFR\
$\vfrs=(\dicts,\xmap,M)$ is then obtained by letting $\dicts=\bigcup_{N\ge 1}
\dictns$. The procedure starts from an empty dictionary, and adds to it
sequences of increasing length $n=1,2,3,\ldots$, sequentially, so that for
each type class $\typet\in\types$, it ``greedily'' augments $\dicts$ with the
largest possible number of sequences in $\typet$ that is a multiple of $M$
and such that these sequences have no prefix in $\dicts$.
The procedure is presented as a characterization of $\vfrs$, rather than as a
computational device. An effective, sequential implementation of $\vfrs$ will
be presented in Subsection~\ref{sec:sequential}.

\begin{theorem}\label{lem:valid}
The \TVFR\ $\vfrns=(\dictns,\xmapst,M)$ constructed by Procedure \GreedyI\ is
universal.
\end{theorem}
\begin{proof}
The fact that the set $\dictns\cup\failns$ constructed by the procedure is
prefix-free and \complete\ can be easily seen by induction in $n$: the
procedure starts with $\dictns\cup\failns=\alp$ and at each iteration it
moves sequences from $\failns$ to $\dictns$ (Step~\ref{step:divide}) and
replaces the remaining sequences of length $n$ in $\failns$ with a full
complement of children of length $n+1$ (Step~\ref{step:stop}). Sequences from
a type class $\typet$ are added to $\dictns$ in sets of size $\jT M$, for
some $\jT\ge 0$ depending on $\typet$ (Step~\ref{step:divide}), and are
assigned uniformly to symbols in $\fixalp$ (Step~\ref{step:map}). Therefore,
by Lemma~\ref{lem:condition-VF}, the constructed \TVFR\ is universal.
\end{proof}

The following key lemma is the basis for the proof of pointwise optimality of
$\vfrns$.

\begin{lemma}\label{lem:mod M}
Let $\dict$ be the dictionary of a universal \VFR. Then, for every type class
$\typet\in\types$, we have
\begin{equation}\label{eq:R pmod M}
|\failn[n](\typet)|\equiv |\typet| \pmod{M}
\end{equation}
and, in particular, for the dictionary $\dicts$, we have
\begin{equation}\label{eq:R bmod M}
|\fails_n (\typet)| = |\typet| \bmod M\,.
\end{equation}
Moreover, if $|\failn[n](\typet)|<M$ for every type class $\typet\in\types$,
$1 \le n \le N$, then $|\dict(\typet)|=|\dicts(\typet)|$ for all
$\typet\in\types$, $1 \le n \le N$.
\end{lemma}

\begin{proof}
Let $T \in \types$ and $T' \in \types[m]$, with $m<n$. For an arbitrary $y^m
\in T'$, consider the set
\begin{equation}\label{eq:diftype}
\Delta(T,T') =
\{\,z^{n-m} \in \alp^{n-m} \,\bigbar \,\, y^m z^{n-m} \in T \,\}
\end{equation}
which, by Fact~\ref{fact:same_sn}, depends only on $T$ and $T'$ and is
independent of the choice of $y^m$ (in fact, $\Delta(T,T') \in \types[n-m]$,
but with an initial state equal to the common final state of the sequences in
$T'$). Now, since $\dict$ is a \complete\ prefix set, if $y^m \in \dict(T')$
then $y^m z^{n-m} \notin \treed(T)$ and, conversely, each $x^n \in T
\setminus
\treed(T)$ must have a unique proper prefix $x^m \in \dict$, in a type class
$T' \in \types[m]$. Since a sequence in $T$ either has a proper prefix in
$\dict$ or it corresponds to a node in $\treed$, in which case the node is
either a leaf (sequences in $\dict(T)$) or internal (sequences in
$\fail_n(T)$), it follows that
\begin{equation}\label{eq:lostkey}
|\failn[n](\typet)| = |T| - |\dict(T)| -
\sum_{m=1}^{n-1} \sum_{T' \in \types[m]} |\Delta(T,T')| \cdot |\dict(T')| \,,
\end{equation}
where the double summation is the number of sequences in $\typet$ that have a
proper prefix in $\dict$. By Lemma~\ref{lem:condition-VF}, each $|\dict(T')|$
in~(\ref{eq:lostkey}), as well as $|\dict(T)|$, must be divisible by $M$,
implying~(\ref{eq:R pmod M}). Equation~(\ref{eq:R bmod M}) then follows from
the construction in Procedure~\GreedyI, which guarantees that
$|\fails_n(\typet)|<M$.

Next, assume $|\dict(\typet)| \neq |\dicts(\typet)|$ for some
$\typet{\in}\types$, $1 \le n \le N$. Without loss of generality, assume $n$
is the smallest such integer, so that $|\dict(T')|{=}|\dicts(T')|$ for all
$T'{\in}\types[m]$, $m<n$. By~(\ref{eq:lostkey}), we have $\fail_n(\typet)
\neq \fails_n(\typet)$. But, if $|\fail_n(\typet)|{<}M$, by~(\ref{eq:R pmod
M}) and~(\ref{eq:R bmod M}), we must have
$|\fail_n(\typet)|{=}|\fails_n(\typet)|$. Therefore, we must also have
$|\dict(\typet)| = |\dicts(\typet)|$ for every $\typet{\in}\types$,
$1{\,\le\,}n{\,\le\,}N$.
\end{proof}

The following theorem establishes the pointwise optimality of $\vfrns$ and
also the uniqueness of the optimal type profile for a universal \VFR.

\begin{theorem}\label{theo:optimal}
Let $\vfr = (\dict,\xmap,M)$ be a universal \VFR. Then, for every $N\ge 0$,
we have $\lenpm(\dictns)\le\lenpm(\dictn)$ and $P(\failns)\le P(\failn)$ for
all $P\in\Pclassk$. Moreover, if $|\dict(\typet)|\ne |\dicts(\typet)|$ for
any $n$ and $\typet \in \types$, then $\lenpm(\dictns)<\lenpm(\dictn)$ for
all $N > n$ and all $P\in\Pclassk$.
\end{theorem}
\begin{proof}
If $\vfr$ is universal, then, by Lemma~\ref{lem:mod M}, we have
$|\failns(\typet)| \le |\failn(\typet)|$ for all $\typet \in \types[N]$ and
thus, since sequences of the same type are equiprobable, $P(\failns) \le
P(\failn)$ for every $N\ge 0$. Moreover, by~(\ref{eq:internal}), we also have
$\lenpm(\dictns) \le
\lenpm(\dictn)$. Now, if there exists $\typet \in \types$ such that
$|\dict(\typet)| \ne |\dicts(\typet)|$ then, by Lemma~\ref{lem:mod M}, we
have $|\fail_{n'}(T')| \ge M$ for some $T' \in \types[n']$, with $n' \le n$.
Therefore, $P(\fails_{n'})<P(\failn[n'])$ which, by~(\ref{eq:internal}),
implies that $\lenpm(\dictns)<\lenpm(\dictn)$ for all $N > n$ and all
$P\in\Pclassk$.
\end{proof}

\begin{remark}\label{rmk:initialVFR}
A modification analogous to the one presented in Remark~\ref{rmk:initialFVR}
is valid in the \VFR\ setting as well for the case in which the initial state
is not deterministic. Specifically, the \VFR\ consumes $k$ input symbols and
then applies $\vfrs$ with initial state $x_1^k$. Equivalently, the dictionary
of the modified procedure corresponds to a tree obtained by taking a
\balanced\ tree of depth $k$, and ``hanging'' from each leaf $s$ the
tree corresponding to $\vfrs$ with initial state $s$. Again, this
\VFR\ is optimal among \VFRs\ that are perfect for every $P \in
\Pclassk$ and every initial state distribution.
\end{remark}

By~(\ref{eq:internal}) and~(\ref{eq:R bmod M}) in Lemma~\ref{lem:mod M}, the
expected dictionary length of $\vfrns$ is given by
\begin{equation}\label{eq:exact}
\lenpm(\dictns) = \sum_{n=0}^{N-1} \sum_{\typet \in \types}
\frac{|\typet| \bmod M}{|\typet|} P(\typet) \,.
\end{equation}
As the exact formula in~(\ref{eq:exact}) appears to provide little insight
into the performance of $\vfrns$ in general (in terms of both expected
dictionary length and failure probability, except for special cases such as
in Example~\ref{ex:ex1} below), a precise estimate is deferred to
subsections~\ref{sec:perpre} and~\ref{sec:perbounds}. In particular, it will
be shown that $P(\failns)$ decays exponentially fast with $N$ and,
consequently, as discussed in Subsection~\ref{sec:definition_vf},
$\lenpm(\dictns)$ converges to the expected dictionary length of $\dicts$,
which is optimal among all universal \VFRs\ with vanishing failure
probability.

\begin{example}\label{ex:ex1}
Consider the VFR $\vfrs$ for the Bernoulli class and $M=3$. Clearly,
$N_0(3)=3$, and there exist two type classes of size $3$ in $\types[3]$,
namely $T = \{ 001,010,100 \}$ and $T' = \{ 011,101,110 \}$. By
Procedure~\GreedyI, both $T$ and $T'$ are included in $\dicts$ and, thus,
$\fails_3 = \{ 000,111 \}$. Next, the procedure considers the set of
extensions of $000$ and $111$. This set does not contain a subset of size $3$
of sequences of the same type for any $n<6$. For $n=6$, four such subsets do
exist, namely, the concatenations of $000$ and $111$ with the sequences in
$T$ and $T'$. Consequently, $\fails_6 = \{ 000000,000111, 111000,111111 \}$
(notice that $|T(000111)|=20$ and there are two sequences of this type in
$\fails_6$, as predicted by~(\ref{eq:R bmod M}) in Lemma~\ref{lem:mod M}).
Again, it can be readily verified that the set of extensions of the sequences
in $\fails_6$ does not contain a subset of size $3$ of sequences of the same
type for $n=7$, but such subsets do exist for $n=8$ (e.g., $\{
00011101,00011110,11100001 \}$). The construction proceeds in a similar
fashion.

Next, we let $P(0){\,=\,}p{\,=\,}1{-}q$ and bound $\lenpm(\dicts)$ in the
example. By~(\ref{eq:exact}), we have
\[
\lenpm(\dicts) = \sum_{n=0}^\infty \sum_{k=0}^n \left (\tbinom{n}{k}
\bmod 3 \right ) p^k q^{n-k} \,.
\]
To derive a lower bound, we take the terms in the sum corresponding to
$k{\,=\,}0,1,2$ and $n{\,-\,}k{\,=\,}0,1,2$, and solve the resulting sums for
$n$ (the cases $k{\,=\,}0,1$ and their complements are straightforward,
whereas for $k=2$ we note that $\tbinom{n}{2}=1\bmod 3$
when $n=2\bmod 3$, and $\tbinom{n}{2} = 0 \bmod 3$ otherwise).\footnote{%
More terms can be estimated using Lucas's Theorem, which applies to any prime
$M$.} After tedious but straightforward computations, it can be shown that
\begin{equation}\label{eq:example}
\lenpm(\dicts) > \frac{1}{pq} - \frac{pq(1+3pq)}{3+p^2q^2} \ge
\frac{1}{pq} - \frac{1}{7} \,.
\end{equation}
Notice that a direct generalization of von Neumann's
scheme~\cite{VonNeumann51} for the case $M=3$ would proceed as
Procedure~\GreedyI\ up to $n=3$, but, in case $x^3 \in \fails_3$, it would
iterate the procedure ``from scratch'' until an output is produced. This
scheme is clearly universal and it is straightforward to show that its
expected dictionary length is $1/(pq)$ which, therefore, by
Theorem~\ref{theo:optimal}, is an upper bound on $\lenpm(\dicts)$. However,
consider the following variant of the iterated scheme: the VFR, denoted
$\notilde{\vfr}=(\notilde{\dict},\notilde{\xmap}, 3)$, outputs
$\notilde{\xmap}(01)=0$, $\notilde{\xmap}(10)=1$, and $\notilde{\xmap}(001) =
\notilde{\xmap}(110)=2$, whereas on inputs $000$ and $111$, the procedure is
iterated. Since $P(01)=P(10)=P(001)+P(110) =pq$, $\notilde{\xmap}$ is
uniformly distributed for all values of $p$. Also, since the expected
dictionary length for the first iteration is $3-2pq$, and an iteration occurs
with probability $p^3+q^3$, overall, we have
\[
\lenpm(\notilde{\dict}) = \frac{3-2pq}{1-p^3-q^3} =
\frac{1}{pq} - \frac{2}{3} < \lenpm(\dicts)
\]
where the inequality follows from~(\ref{eq:example}).
The reason $\notilde{\vfr}$ can outperform $\vfrs$ is that it {\em does
not\/} preserve universality under truncation (note that the condition of
Lemma~\ref{lem:condition-VF} is not satisfied for $\notilde{\vfr}$, and the
truncated scheme $\notilde{\vfr}_N$ is not perfect whenever $N\equiv
2\pmod{3}$). Perfect VFRs without the truncation requirement are studied
in~\cite{Hoeffding-Simons70,Stout-Warren84}: in particular, it is shown
in~\cite{Stout-Warren84} that no single VFR can be optimal, in this sense,
for all values of $p$. Notice also that using the upper bound on
$\lenpm(\dicts)$ we obtain $\lenpm(\dicts)-\lenpm(\notilde{\dict})<2/3$. In
fact, as will be discussed in Section~\ref{sec:perbounds}, as $M$ grows, the
cost of requiring perfection under truncation becomes negligible.
\end{example}

The performance analysis in subsections~\ref{sec:perpre}
and~\ref{sec:perbounds}, as well as the design of an efficient implementation
in Subsection~\ref{sec:sequential}, require the discussion of additional
properties of type classes, dictionaries, and their interactions. We present
these properties in the next subsection.

\subsection{Additional properties}\label{sec:properties}
We next describe a decomposition of a type class $\typet\in\types$. For
$u\in\alp^\ell$, $0\le \ell\le n$, let
\begin{equation}
\label{eq:typenotcut}
\typecutnot[\typet]{u} = \{\,x^{n}\,|\, x^n \in \typet\,,\;x_{n-k-\ell+1}^{n-k}
= u\,\}\,
\end{equation}
(by our assumptions on initial states, and by the range defined for $\ell$,
$\typecutnot[\typet]{u}$ is well defined even if some of the indices of
$x_{n-k-\ell+1}^{n-k}$ in~(\ref{eq:typenotcut}) are negative). Clearly, we
can decompose $T$ as
\begin{equation}\label{eq:typesplit2}
\typet = \mbox{$\bigcup$}_{u\in\alp^\ell} \typecutnot[\typet]{u}\,,
\end{equation}
where, by~(\ref{eq:typenotcut}) and Fact~\ref{fact:same_sn}, the sequences in
$\typecutnot{u}$ coincide in their last $k+\ell$ symbols. For
$u\in\alp^\ell$, define
\begin{equation}\label{eq:typecut}
\typecut{u}  =
\{\,x^{n-\ell}\,|\, x^n \in \typecutnot{u}\,\}\,,\quad
0 \le \ell \le n\,,
\end{equation}
and $\typecut{u}  = \emptyset,\;\ell > n$. From the
definitions~(\ref{eq:typenotcut}),~(\ref{eq:typecut}), it is readily verified
that, for $u,v\in\alp^\ast$,
\begin{equation}\label{eq:au}
\typecut{vu} = \typecut[\typecut{u}]{v} \,.
\end{equation}

What makes the sets $\typecut{u}$ useful is the fact that they are,
generally, type classes themselves, as established in the following lemma.
\begin{lemma}\label{lem:typesplit}
If $\typecut{u}$ is not empty then $\typecut{u} \in \types[n-\ell]\,$.
\end{lemma}
\begin{proof}
We prove the result by induction on $\ell$. For $\ell=0$, the claim is
trivial, since $\typecut{\lambda} = \typet$. Assume the claim is true for all
$\ell'$, $0 \le \ell' < \ell$, and consider a string $u = a u'$, $a\in\alp$,
$u'\in\alp^{\ell-1}$. If $\typecut{u}$ is not empty, then neither is
$\typecut{u'}$, and, by the induction hypothesis, we have $\typecut{u'}=
\typet'\in\types[m]$, where $m = n-\ell+1$. Consider a sequence $x^{m}\in
\typet'$. The type of sequences in $\typecut{u}$ differs from that of $x^{m}$
in one count of $x_{m}$, which is deducted from $\nsa[x_{m}](x^{m})$,
$s=ax_{m-k+1}^{m-1}$, if $k>0$, or from the global count of $x_{m} = a$ if
$k=0$. In either case, by Fact~\ref{fact:same_sn} and the definition of
$\typecut{u}$, both $s$ and $x_{m}$ are invariant over $\typecut{u}$, and,
therefore, $\typecut{u}\subseteq\typet''$ for some type class
$\typet''\in\types[m-1]$. On the other hand, if a sequence
$y^{m-1}\in\typet''$ is extended with a symbol $y_{m}$ (whether $y_{m}{=}a$
when $k=0$ or $y_{m}$ is the invariant final symbol $x_m$ of sequences in
$\typet'$ when $k>0$), then the counts of $\typet'$ are restored, so we have
$y^{m}\in\typecutnot[T']{a}$, and, hence,
$\typet''\subseteq\typecut[T']{a}=\typecut{u}$, where the last equality
follows from~(\ref{eq:typecut}). Therefore, $\typecut{u} = \typet''$.
\end{proof}

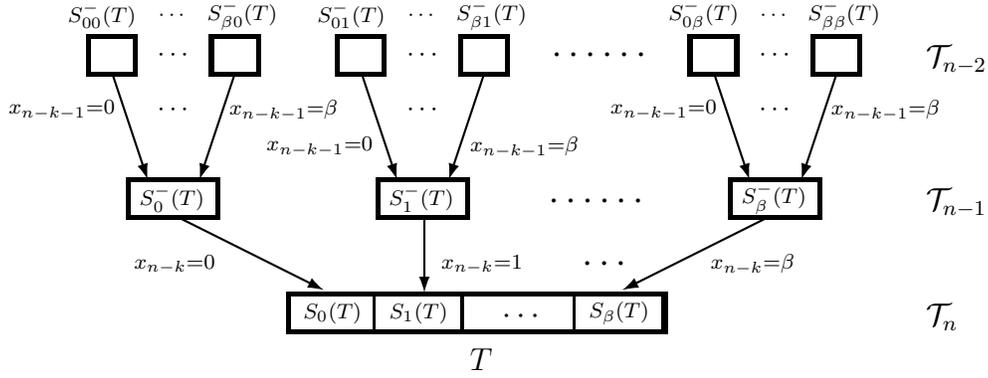
\begin{figure*}
\begin{center}
\setlength{\unitlength}{0.015in}
\setlength{\fboxsep}{0pt}
\setlength\fboxrule{1.5pt}
\newlength{\savefboxrule}
\newlength{\thinboxrule}
\setlength{\thinboxrule}{0.8pt}
\newcommand{\thinbox}{\setlength{\savefboxrule}{\fboxrule}\setlength{\fboxrule}{\thinboxrule}}
\newcommand{\restorebox}{\setlength{\fboxrule}{\savefboxrule}}
\newcommand{\boxheight}{12}
\newcommand{\boxwidth}{30}
\newcommand{\tqboxwidth}{15}
\newcommand{\hboxwidth}{10}
\newcommand{\dboxwidth}{60}
\newcommand{\tboxwidth}{10}
\newcommand{\onebox}{\makebox(\boxwidth,\boxheight)}
\newcommand{\halfbox}{\makebox(\hboxwidth,\boxheight)}
\newcommand{\tqbox}{\makebox(\tqboxwidth,\boxheight)}
\newcommand{\dblbox}{\makebox(\dboxwidth,\boxheight)}
\newcommand{\tinybox}{\makebox(\tboxwidth,\boxheight)}
\newcommand{\fcgray}[1]{\fcolorbox{black}{gray!50}{#1}}
\newcommand{\fcggray}[1]{\thinbox\fcolorbox{black}{gray!50}{#1}\restorebox}
\newcommand{\fcgwhite}[1]{\thinbox\fcolorbox{black}{white}{#1}\restorebox}
\newcommand{\fcwhite}[1]{\fcolorbox{black}{white}{#1}}
\begin{picture}(340,130)(-20,-15)
\thicklines
\put(37,90){
	\put(-30,00){\fcwhite{\tqbox{}}}
	\put(-24.5,20){\makebox(0,0){\footnotesize$\typecut{00}$}}
	\put(-21,00){\vector(1,-3){12.5}}
	\put(0.5,20){\makebox(0,0){\footnotesize$\cdots$}}
	\put(01,06){\makebox(0,0){$\cdots$}}
	\put(13,00){\fcwhite{\tqbox{}}}
	\put(24.5,20){\makebox(0,0){\footnotesize$\typecut{\beta 0}$}}
	\put(22,00){\vector(-1,-3){12.5}}
	\put(-75,-13){
	    \put(55,0){\makebox(0,0)[r]{\footnotesize$x_{n{-}k{-}1}{=}0$}}
	    \put(76,00){\makebox(0,0){$\cdots$}}
	    \put(95,0){\makebox(0,0)[l]{\footnotesize$x_{n{-}k{-}1}{=}\beta$}}
	}
}
\put(124,90){
	\put(-30,00){\fcwhite{\tqbox{}}}
	\put(-24.5,20){\makebox(0,0){\footnotesize$\typecut{01}$}}
	\put(-21,00){\vector(1,-3){12.5}}
	\put(0.5,20){\makebox(0,0){\footnotesize$\cdots$}}
	\put(01,06){\makebox(0,0){$\cdots$}}
	\put(13,00){\fcwhite{\tqbox{}}}
	\put(24.5,20){\makebox(0,0){\footnotesize$\typecut{\beta1}$}}
	\put(22,00){\vector(-1,-3){12.5}}
	\put(-75,-26){
	    \put(58,0){\makebox(0,0)[r]{\footnotesize$x_{n{-}k{-}1}{=}0$}}
	    \put(76,13){\makebox(0,0){$\cdots$}}	
        \put(92,0){\makebox(0,0)[l]{\footnotesize$x_{n{-}k{-}1}{=}\beta$}}
	}
}
\put(187,96){\makebox(0,0){\Large$\cdots\cdots$}} %
\put(247,90){
	\put(-30,00){\fcwhite{\tqbox{}}}
	\put(-24.5,20){\makebox(0,0){\footnotesize$\typecut{ 0\beta}$}}
	\put(-21,00){\vector(1,-3){12.5}}
	\put(01,06){\makebox(0,0){$\cdots$}}
	\put(0.5,20){\makebox(0,0){\footnotesize$\cdots$}}
	\put(13,00){\fcwhite{\tqbox{}}}
	\put(24.5,20){\makebox(0,0){\footnotesize$\typecut{\beta\beta}$}}
	\put(22,00){\vector(-1,-3){12.5}}
	\put(-75,-13){	    \put(55,0){\makebox(0,0)[r]{\footnotesize$x_{n{-}k{-}1}{=}0$}}
	\put(76,00){\makebox(0,0){$\cdots$}}	
    \put(95,0){\makebox(0,0)[l]{\footnotesize$x_{n{-}k{-}1}{=}\beta$}}
	}
}
\put(65,40){ %
  \put(-44,0)  {
	\put(00,00){\fcwhite{\onebox{}}}
	\put(16,06){\makebox(0,0){\footnotesize$\typecut{0}$}}
  }
  \put(43,0) {
	\put(00,00){\fcwhite{\onebox{}}}
	\put(16,06){\makebox(0,0){\footnotesize$\typecut{1}$}}
  }
  \put(166,0){
	\put(00,00){\fcwhite{\onebox{}}}%
	\put(16,06){\makebox(0,0){\footnotesize$\typecut{\beta}$}}
  }
  \put(121,05){\makebox(0,0){\Large$\cdots\cdots$}}
}
\put(20,-50){
\put(020,89){\vector(2,-1){50}}
\put(105,89){\vector(0,-1){26}}
\put(224,89){\vector(-2,-1){50}}
\put(0,73){
    \put(32,0){\makebox(0,0)[r]{\footnotesize$x_{n-k}{=}0$}}
    \put(111,0){\makebox(0,0)[l]{\footnotesize$x_{n-k}{=}1$}}
    \put(160,0){\makebox(0,0)[l]{\Large$\cdots$}}
    \put(205,0){\makebox(0,0)[l]{\footnotesize$x_{n-k}{=}\beta$}}
}
}
\put(77,00){
	\multiput(0,0)(\boxwidth,0){2}{\fcgwhite{\onebox{}}}
	\put(16,06){\makebox(0,0){\footnotesize$\typecutnot{0}$}}
	\put(46,06){\makebox(0,0){\footnotesize$\typecutnot{1}$}}
	\put(82.5,5){\makebox(0,0){\large$\cdots$}}
	\put(100,0){\fcgwhite{\onebox{}}}%
	\put(116,06){\makebox(0,0){\footnotesize$\typecutnot{\beta}$}}
	\fbox{\makebox(130,\boxheight){}} %
	\put(-65,-10){\makebox(0,0){\large$\typet$}}
}
\put(300,04.5){\makebox(0,0)[l]{\large$\types[n]$}}
\put(300,44.5){\makebox(0,0)[l]{\large$\types[n-1]$}}
\put(300,94.5){\makebox(0,0)[l]{\large$\types[n-2]$}}
\end{picture} %
\end{center}
\caption{\label{fig:typerelations}Type class relations
(with $\alp=\{0,1,\ldots,\beta\},\;\;\beta=\alpha{-}1$).}
\end{figure*}

\begin{remark} \label{rem:distinct}
When $\ell > k+1$, the type classes $\typecut{u}$ and $\typecut{u'}$ may
coincide even if $u \neq u'$. Specifically, letting $s_f \in \alp^k$ denote
the final state of $T$, this situation arises if and only if $u = u_1^k v$,
$u' = u_1^k v'$, and $T(vs_f)=T(v's_f)$, where both type classes assume an
initial state $u_1^k$ (it is easy to see that this situation requires
$|v|>1$). In fact, the type class $T(vs_f)$ is precisely the set
$\Delta(T,T')$ defined in~(\ref{eq:diftype}), for $T'=\typecut{u}$.
\end{remark}

Equations~(\ref{eq:typesplit2})--(\ref{eq:typecut}) and
Lemma~\ref{lem:typesplit} show how we can trace the origins of sequences in a
type class $\typet$ to the type classes $\typecut{u}$ of their prefixes.
This relation between type classes is illustrated, for $\ell=2$, in
\figurename~\ref{fig:typerelations}. Extending the figure recursively,
using~(\ref{eq:au}), $\typet$ can be seen as being at the root of an
$\alpha$-ary tree tracking the path, through shorter type classes, of
sequences that end up in $\typet$. This structure will be useful in the
derivation of various results in the sequel.

We now apply the foregoing type class relations to obtain a recursive
characterization of $\fail_n(\typet)$, $\typet\in\types$, for a given
dictionary $\dict$.
\begin{lemma}\label{lem:A from R}
For any dictionary $\dict$ and any class type $\typet\in\types$, there is a
one-to-one correspondence between $\treed(\typet)$ and $\bigcup_{a \in \alp}
\idict(\typecut{a})$, which implies
\begin{equation}\label{eq:A from R num}
|\fail_n(\typet)| =
\sum_{a \in \alp} |\failn[n-1](\typecut{a})| - |\dict(\typet)| \,,
\end{equation}
where we take $\failn[n-1](\emptyset) = \emptyset$ (some of the sets
$\typecut{a}$ may be empty).
\end{lemma}

\begin{proof}
Clearly, by the decomposition~(\ref{eq:typesplit2}), each node in $\treed(T)$
is a child of a node in $\idict(\typecut{a})$ for some $a \in \alp$.
Conversely, a node in $\idict(\typecut{a})$ has a \emph{unique} child in
$\treed(T)$, in the direction prescribed by the decomposition. To complete
the proof, we recall that, for any $T'\in\types[n']$ and any $n'$, the nodes
in $\idict(T')$ correspond to the sequences in $\fail_{n'}(T')$, whereas the
sequences in $\dict(T')$ correspond to the leaves in $\treed(T')$.
\end{proof}

\subsection{Sequential implementation of $\vfrs$}
\label{sec:sequential}
\newcommand{\GreedyII}{G2}
Procedure \GreedyI\ in \figurename~\ref{figproc:greedy} constructs
dictionaries of depth $N$, for arbitrarily large values of $N$. In a
practical \VFR\ application, however, dictionaries are not actually
constructed. Instead, what is required is a procedure that reads the input
sequence $x_1x_2 \ldots x_n \ldots$, sequentially, and, for each $n$, makes a
decision as to whether it should produce an output (and what that output
should be) and stop, or continue processing more input. Procedure \GreedyII\
in \figurename~\ref{figproc:sequential} describes such a sequential
implementation of $\vfrs$, without truncation. The procedure can easily be
modified to implement a \TVFR\ for arbitrary $N$, with a possible failure
exit.

\begin{figure}
\begin{figprocedure}
{Sequence $x_1 x_2 x_3 \cdots x_n \cdots $, integer $M > 1$.}
{Number $\rnd\in[M]$.}

\item Set $\indx_\remainltr = 0,\; n=0,\;\tcounts(x^n) = \mathbf{0}$.
\item\label{step:read} Increment $n$, read $x_n$ and update $\tcounts(x^n)$.
    Let $T = \type$.
\item\label{step:RD} Compute $|\fail_{n-1}(\typecut[\typet]{a})| =
    |\typecut[\typet]{a}| \bmod M$, for each $a\in\alp$.
\item\label{step:IA} Set $\indx_{\treed} = \sum_{a<x_{n-k}}
    |\fail_{n-1}(\typecut[\typet]{a})| + \indx_\remainltr$.
\item\label{step:jT} Set $\jT = \Bigl\lfloor \sum_{a\in\alp}
    |\fail_{n-1}(\typecut{a})|/M \Bigr\rfloor$.

\item\label{step:decide} If
$\indx_{\treed}{<\,}\jT M$ then \textbf{output}
    $\rnd{\,=\,}\indx_{\treed}{\,\bmod\,}M$ and \textbf{stop}.

Otherwise, set $\indx_\remainltr = \indx_{\treed} - \jT M$ and \textbf{go to}
Step~\ref{step:read}.
\end{figprocedure}
\caption{\label{figproc:sequential}Procedure \GreedyII: Sequential
implementation of $\vfrns$.}
\end{figure}

The procedure relies on a sequential alphabetic enumeration of
$\treed(\type)$, which defines a partition of this set into $\fail_n(\type)$
and $\dict(\type)$, and yields an enumeration of the two parts. These
enumerations are based, in turn, on the one-to-one correspondence of
Lemma~\ref{lem:A from R}, and determine whether an output is produced, and
the value of the output. We assume a total (alphabetic) order $<$ of the
elements of $\alp$; for the purpose of comparing sequences of length $n$,
significance increases with the coordinate index (i.e., $x_n$ is the most,
and $x_1$ the least, significant symbol in $x^n$). We assume, recursively,
that after processing $x^{n-1}$ we have its index
$\indx_{\remainltr}(x^{n-1})$ in $\fail_{n-1}(\typecut[\type]{x_{n-k}})$
(notice that $\type[x^{n-1}] =
\typecut[\type]{x_{n-k}}$). In \figurename~\ref{figproc:sequential}, this
index is assumed stored in the variable $\indx_\remainltr$ when
Step~\ref{step:read} is reached. Since all the sequences in $\typet = \type$
coincide in their last $k$ symbols, if $y^n\in\typet$ and $y_{n-k} <
x_{n-k}$, then $y^n < x^n$ in the alphabetical order. Therefore, using the
one-to-one correspondence of Lemma~\ref{lem:A from R}, the index of $x^n$ in
$\treed(\typet)$ is given by
\begin{equation}\label{eq:IA}
\indx_{\treed}(x^n) =
\sum_{a<x_{n-k}} |\fail_{n-1}(\typecut[\typet]{a})|
+ \indx_\remainltr(x^{n-1})\,.
\end{equation}
In \figurename~\ref{figproc:sequential}, this computation is performed in
Step~\ref{step:IA}, based on the value of the aforementioned index
$\indx_\remainltr(x^{n-1})$ available when Step~\ref{step:read} was reached,
and on the values $|\fail_{n-1}(\typecut[\typet]{a})|$ computed in
Step~\ref{step:RD}. The computations in~Step~\ref{step:RD}, in turn, can be
derived from~(\ref{eq:R bmod M}), by means of Whittle's
formula~\cite{Whittle} applied to $\typecut[\typet]{a} \in \types[n-1]$.
Notice that the type $\tcounts^{(a)}$ associated with $\typecut[\typet]{a}$,
which is required to evaluate Whittle's formula, is easily obtained from the
type $\tcounts(x^n)$. In Step~\ref{step:jT}, the factor $\jT$ of $|\dict(T)|
= \jT M$ is obtained, based again on the quantities computed in
Step~\ref{step:RD} and on~(\ref{eq:A from R num}). We assume that
$\dict(\typet)$ consists of the first $\jT M$ sequences in the alphabetic
ordering of $\treed(\typet)$. If the index~(\ref{eq:IA}) of $x^n$ in this
ordering is less than $\jT M$, as checked in Step~\ref{step:decide}, then
$x^n$ is in the dictionary, and an output is produced. Setting the output
value as in Step~\ref{step:decide} guarantees that sequences in
$\dict(\typet)$ are assigned uniformly to values in $\fixalp$, as required by
the condition of Lemma~\ref{lem:condition-VF}. If $x^n$ is not in $\dict$,
then its index in $\fail_n(\typet)$ is obtained by subtracting
$|\dict(\typet)| = \jT M$ from its index in $\treed(\typet)$, so
$\fail_{n}(\typet)$ inherits the alphabetic ordering of
$\treed(\typet){\,=\,}\dict(\typet){\,\cup\,}\fail_n (\typet) $, and the
assumptions for the next iteration are satisfied. Since Whittle's formula can
be evaluated in time polynomial in $n$, the procedure in
\figurename~\ref{figproc:sequential} runs in polynomial time.

\subsection{Performance: Preliminaries}
\label{sec:perpre}
We study the performance of $\vfrns$ in terms of expected dictionary length
and failure probability for large $N$. First, we show that the failure
probability vanishes exponentially fast and that, as a result, the expected
dictionary length converges. We then characterize the asymptotic behavior
with respect to $M$ of the convergence value, up to an additive constant
independent of $M$. To this end, for sufficiently large $N$, we derive a
lower bound on $\lenpm(\dictn)$ for any \truncatedD\ dictionary $\dictn$
derived from a universal \VFR\ and we show that the bound is achievable
within such a constant. For the achievability result we will not use the
optimal universal \VFR\ $\vfrs$, but a different \VFR, for which the analysis
is simpler.\footnote{The situation is akin to lossless source coding, for
which the entropy bound is shown to be achievable with, say, the Shannon
code, rather than with the (optimal) Huffman code.}

\begin{theorem}\label{theo:failure}
For every $P \in \Pclassk$, $P(\failns)$ decays exponentially fast with $N$.
\end{theorem}

\begin{proof}
By~(\ref{eq:R bmod M}), recalling that sequences of the same type are
equiprobable, for any $\epsilon > 0$ we have
\begin{equation}\label{eq:expfail}
P(\failns) <
P \left ( |T(X^N)| < 2^{N(\Hrate-\epsilon)} \right ) +
M 2^{-N(\Hrate-\epsilon)} \,.
\end{equation}
Now, recalling that, if the maximum-likelihood estimator for $x^N$ is bounded
away from the boundary of $\pDomain$, then $|T(x^N)|$ is exponential in $N
\hat{H}_k (x^N)$ for large $N$, the event $|T(X^N)| < 2^{N(\Hrate -
\epsilon)}$ is a large deviations one. Therefore, both terms on the rightmost
side of~(\ref{eq:expfail}) decay exponentially fast with $N$.
\end{proof}

As argued, by~(\ref{eq:internal}), Theorem~\ref{theo:failure} implies that
$\lenpm(\dictns)$ converges to the expected dictionary length of $\dicts$.
Next, we develop the basic tools we will use in the characterization of
$\lenpm(\dictns)$, starting with some definitions. Throughout this subsection
we assume that all dictionaries satisfy that, for every $P \in \Pclassk$,
$P(\failn)$ vanishes as $N \to \infty$. In the next subsection, as we apply
these tools to specific dictionaries, this property will need to be verified
in each case.

The entropy of a dictionary $\dict$ is defined as $\Hdict \defined
-\Exp_{P,\dict} \log P(X^\ast)$. When the initial state $s$ differs from
$s_0$, we use the notation $\Hdict[P,s]$ for this entropy. As shown
in~(\ref{eq:exptruncf}), due to the vanishing failure probability assumption,
if $\dict$ is infinite then $\Hdict$ coincides with the limit of the entropy
of the truncated dictionary, completed with the failure set. For the latter
entropy, we use the notation $\Hdictn$, omitting the union with the failure
set (just as in $\lenpm(\dictn)$).

Given $\delta > 0$, let
\begin{equation}
\label{eq:deltacounts}
S^{(\delta)}_n \defined \{\,x^n \in \alp^n \,\bigbar
\exists \,\, a{\,\in\,}\alp,\,s{\,\in\,}\alp^k \,\text{ s.t. }\,\nsa(x^n) <
\delta n \, \} \,.
\end{equation}
Thus, by our assumption that all conditional probabilities in processes $P
\in \Pclassk$ are nonzero, for sufficiently small $\delta$ (depending on the
specific $P$), $S^{(\delta)}_n$ is a large deviations event and thus its
probability vanishes exponentially fast with $n$.

Our results are based on the following key lemma on dictionary entropies,
where we use bounding techniques that are rooted in the source coding
literature, particularly~\cite{tjalkens87} and~\cite{tjalkens92}. In the
sequel, the $O(\cdot)$ notation refers to asymptotics relative to $M$. Thus,
$O(1)$ denotes a quantity whose absolute value is upper-bounded by a
constant, independent of $M$.

\begin{lemma}\label{lem:asymptotics}
Let $P \in \Pclassk$ and let $\dict$ be a dictionary.
\begin{itemize}
\itemsep=-0.2mm
\item [(i)]
If for every $\xs \in \dict$ we have $|T(\xs)| \ge M$, then
\begin{equation}\label{eq:lower-HD}
\Hdictn \ge \log M + (K/2) \log \log M - O(1)
\end{equation}
for every $N>(\log M)/(\Hrate-\epsilon)$, where $\epsilon > 0$ is
arbitrary. In addition,~(\ref{eq:lower-HD}) also holds in the limit for
$\Hdict$.
\item [(ii)]
If for every $\xs \in \dict$ we have $|T(\xs)| < CM$ for some positive
constant $C$, except for a subset $\dict_0$ of $\dict$ such that
\begin{equation}\label{eq:subcond}
\sum_{\xs \in \dict_0} |\xs| P(\xs) = O(1) \,,
\end{equation}
then
\begin{equation}\label{eq:HD}
\Hdictn \le \log M + (K/2) \log \log M + O(1)
\end{equation}
for every $N > 0$. In addition,~(\ref{eq:HD}) also holds in the limit for
$\Hdict$.
\end{itemize}
\end{lemma}

\begin{proof}
Consider the universal sequential probability assignment $Q(\xs)$ on
$\alp^\ast$ given by a uniform mixture over $\Pclassk[k]$ (namely, Laplace's
rule of succession applied state by state), for which it is readily verified
that
\[
Q(x^n) = \prod_{s \in \alp^k} \frac{\prod_{a \in \alp} \nsa(x^n)! }{(\ns(x^n)+\alpha-1)! }
(\alpha-1)! \,.
\]
It follows from Whittle's formula~\cite{Whittle} that
\begin{equation} \label{eq:mixture}
Q(x^n) = \frac{W(x^n)} {|T(x^n)|}
\prod_{s \in \alp^k} \tbinom{\ns(x^n)+\alpha-1}{\alpha-1}^{-1}
\end{equation}
where $W(x^n)$ denotes a determinant in the formula (``Whittle's cofactor'')
that satisfies $0 <  W(x^n) \le 1$, and accounts for certain restrictions in
the state transitions which limit the universe of possible sequences that
have the same type as $x^n$. We use this probability assignment to write the
entropy of a generic, finite dictionary $\dict'$, as
\begin{equation} \label{eq:divergence}
H_P (\dict') = - \Exp_{P,\dict'} \left [ \log \frac{P(X^\ast)}{Q(X^\ast)} \right ]
+ \Exp_{P,\dict'} \left [ \log\frac{1}{Q(X^\ast)} \right ] ,
\end{equation}
where the first summation is the divergence between $P$ and $Q$ as
distributions over $\dict'$, which we denote $D_{\dict'} (P||Q)$.
By~(\ref{eq:mixture}), we obtain
\begin{align}
H_P (\dict') = & - D_{\dict'} (P||Q) \nonumber\\ %
& + \Exp_{P,\dict'} \biglbracket[1.5em]{[} \log |T(X^\ast)| -
\log W(X^\ast)\nonumber\\
&\hspace{5em} + \sum_{s \in \alp^k} \log \tbinom{\ns(X^\ast)+\alpha-1}{\alpha-1}
\bigrbracket[1.5em]{]}  \,.
 \label{eq:divergence2}
\end{align}

To prove Part~(i), we first notice that if $\dict''$ is also a dictionary and
$\tree_{\dict''} \subseteq \tree_{\dict'}$, then it is easy to see that $H_P
(\dict'') \le H_P (\dict')$. Therefore, it suffices to prove the lemma for $N
= N_1$, where
\begin{equation} \label{eq:N1}
N_1 \defined \left \lceil \frac{\log M}{\Hrate-\epsilon} \right \rceil \,.
\end{equation}
For convenience, the dictionary $\dictn[N_1] \cup \failn[N_1]$ is denoted
$\dict'$. Now, if $x^n \in \dict'$, then either $|T(x^n)| \ge M$ or $n =
N_1$, implying that
\begin{align}
\Exp_{P,\dict'} & \log |T(X^\ast)| \nonumber\\
& \ge
\left [1 - P \left (|T(x^{N_1})| < 2^{N_1(\Hrate-\epsilon)}
\right ) \right ] \log M \,.
\label{eq:event}
\end{align}
Since the probability on the right-hand side of~(\ref{eq:event}) decays
exponentially fast with $N_1$ (see proof of Theorem~\ref{theo:failure}), we
obtain
\begin{equation} \label{eq:LBsize}
\Exp_{P,\dict'} \log |T(X^\ast)| \ge \log M - O((\log M)/M) \,.
\end{equation}

Next, recalling the definition of $N_0(M)$ in~(\ref{eq:N0}), for every $\xs
\in \dict'$ and $s \in \alp^k$, $x^{N_0(M)}$ is a prefix of $\xs$, and
therefore $\ns(\xs) \ge \ns(x^{N_0(M)})$. If $x^{N_0(M)} \notin
S^{(\delta)}_{N_0(M)}$ (recall~(\ref{eq:deltacounts})), we have
$\ns(x^{N_0(M)}) > \alpha \delta N_0(M) \ge(\alpha\delta\log
M)/(\log\alpha)$. Thus, by Stirling's approximation, sequences $x^{N_0(M)}
\in \alp^{N_0(M)}{\,\setminus\,}S^{(\delta)}_{N_0(M)}$ satisfy
\begin{equation} \label{eq:binobound}
\sum_{s \in \alp^k} \log \tbinom{\ns(x^{N_0(M)})+\alpha-1}{\alpha-1} >
K \log \log M - O(1) \,.
\end{equation}
Since $N_0(M)=\Omega(\log M)$ we have, for sufficiently small $\delta$,
$P(S^{(\delta)}_{N_0(M)})=O(1/M)$. Thus, the right-hand side
of~(\ref{eq:binobound}) is also a lower bound on $\Exp_{P,\dict'}
\sum_{s \in \alp^k} \log \tbinom{\ns(X^\ast)+\alpha-1}{\alpha-1}$.
Finally, we have $\log W(\xs) \le 0$ for every $\xs \in \dict'$, so we
conclude from~(\ref{eq:divergence2}), (\ref{eq:LBsize}),
and~(\ref{eq:binobound}), that
\begin{equation} \label{eq:divergence3}
H_P (\dict') \ge - D_{\dict'} (P||Q) + \log M + K \log \log M - O(1) \,,
\end{equation}
where the $O(1)$ term depends on $P$.

As for the divergence term in~(\ref{eq:divergence3}), it is easy to see that,
since $\tree_{\dict'}$ is a sub-tree of the \balanced\ tree of depth $N_1$,
we have
\begin{equation} \label{eq:extension}
D_{\dict'} (P||Q) \le D_{\alp^{N_1}} (P||Q) \,.
\end{equation}
Applying the divergence estimate in~\cite{atteson} (which extends to Markov
sources the results in~\cite{clarke} on the asymptotics of the redundancy of
Bayes rules with continuous prior densities) to sequences of length $N_1$, we
conclude that the divergence term in~(\ref{eq:divergence3}) is upper-bounded
by $(K/2)\log \log M + O(1)$, proving (\ref{eq:lower-HD}).\footnote{While the
claim in~\cite[Corollary 1]{atteson} is for a source in stationary mode, its
proof actually builds on showing the same result for a fixed initial state,
as in our setting.} The claim on $\Hdict$ follows from its definition as the
limit of a nondecreasing sequence.

To prove Part~(ii) we first notice that, since $\Hdictn$ is nondecreasing in
$N$, it suffices to prove the upper bound for sufficiently large $N$.
Moreover, since extending dictionary strings by a finite amount cannot lower
the entropy, it suffices to prove it for dictionaries $\dict$ such that the
length of the shortest sequence in the dictionary is at least $N_2 \defined c
\log M$ for some positive constant $c$. Therefore, proceeding as
in~(\ref{eq:extension}) we have, for sufficiently large $N$,
\begin{equation} \label{eq:divlow}
D_{\dictn \cup \failn} (P||Q) \ge D_{\alp^{N_2}} (P||Q) =
\frac{K}{2} \log \log M + O(1) \,,
\end{equation}
where the estimate follows, again, from~\cite[Corollary 1]{atteson}.

Next, we observe that, for $x^n \in \dict \setminus \dict_0$, we have
$|T(x^n)| < CM$, whereas for $x^n \in \dict_0$ we can use the trivial bound
$\log |T(x^n)| < n \log \alpha$. Therefore, by~(\ref{eq:subcond}),
\begin{equation} \label{eq:typeup}
\Exp_{P,\dict} \log |T(X^\ast)| < \log M + O(1) \,.
\end{equation}
Similarly, for $n>N_1(C) \defined (\log (CM))/(\Hrate-\epsilon)$ for some
$\epsilon >0$, $\{x^n \in \dict \setminus \dict_0\}$ is a large deviations
event and its probability decreases exponentially fast with $n$, as the type
class size for a typical sequence will be at least $CM$. Thus, using again
Stirling's approximation and~(\ref{eq:subcond}), we obtain
\begin{align}
\sum_{s \in \alp^k} \Exp_{P,\dict} \log \tbinom{\ns(X^\ast)+\alpha-1}{\alpha-1}
&< K \log N_1(C) + O(1) \nonumber\\
& = K \log \log M + O(1).
\label{eq:binup}
\end{align}
Clearly, since $P(\failn)$ vanishes as $N$ grows, the upper
bounds~(\ref{eq:typeup}) and~(\ref{eq:binup}) hold, \emph{a fortiori}, when
the expectations are taken over sequences in $\dictn \cup \failn$ instead of
$\dict$, for any $N>0$.

Finally, for sequences $x^n \in \alp^n \setminus S^{(\delta)}_n$ for some
$\delta>0$, $W(x^n)$ is known to be lower-bounded by a positive function of
$\delta$ (see, e.g.,~\cite[proof of Lemma 3]{WMF94}). For sequences $x^n \in
S^{(\delta)}_n$ (an event whose probability decreases with $n$ exponentially
fast for sufficiently small $\delta$), $W(x^n)$ is
$\Omega(1/n^k)$~\cite{Boza}. Therefore, $\Exp_{P,\dictn\cup\failn}\log
(1/W(X^\ast))=O(1)$ for any $N>0$. The upper bound~(\ref{eq:HD}) then follows
from~(\ref{eq:divergence2}), (\ref{eq:divlow}), (\ref{eq:typeup}),
and~(\ref{eq:binup}), both for $\Hdictn$ and for $\Hdict$.
\end{proof}

To apply Lemma~\ref{lem:asymptotics} to the estimation of the expected
dictionary length, we need to link $\lenp(\dictn)$ to $\Hdictn$. Applying the
LANSIT (recall~(\ref{eq:lansit})) to the self-information function, we obtain
the ``leaf-entropy theorem'' (see, e.g.,~\cite{massey83}), which states, for
a generic dictionary $\dict$, that
\[
\Hdict = \sum_{\xs \in \idict} P(\xs) H_P (X|s(s_0,\xs)) \,,
\]
where $s(s_0,\xs)$ denotes the state assumed by the source after emitting
$\xs$, starting at $s_0$. In the memoryless case, $H_P (X|s(s_0,\xs))$ is
independent of $\xs$, and further applying the LANSIT to the length function
(as in~(\ref{eq:internal})), we obtain $\Hdict = \Hrate \lenp(\dict)$. This
relation directly provides the desired link, and is used, e.g.,
in~\cite{HanHoshi97}. The situation is more intricate for sources with
memory, for which, regrouping terms, the theorem clearly takes the form
\[
\Hdict = \sum_{t \in \alp^k} H_P (X|t) \sum_{\xs \in \idict}
P(\xs) \delta(t,s(s_0,\xs)) \,,
\]
where $\delta(t,s(s_0,\xs))=1$ if $t=s(s_0,\xs)$, and $0$ otherwise. An
additional application of the LANSIT, this time to the function $n_t (\xs)$,
then yields
\begin{equation} \label{eq:leaf-entropy}
\Hdict = \sum_{t \in \alp^k} H_P (X|t) \Exp_{P,\dict} \;n_t (X^\ast) \,.
\end{equation}

A variant of this problem is studied in~\cite{tjalkens87} in the context of
variable-to-fixed-length source coding. We will make use of a result
in~\cite{tjalkens87}, for which we need to consider the \emph{extended}
source defined on the strings of $\dict$ (referred to as a ``segment source''
in~\cite{tjalkens87}), which is clearly also Markov with the same state set
$\alp^k$. Let $P^{\rm{seg}}(s)$ denote its steady-state distribution when the
chain is started with the stationary distribution of the basic (non-extended)
source.\footnote{As noted in~\cite{tjalkens87}, the segment source may not be
irreducible, and therefore $P^{\rm{seg}}(s)$ is one of possibly multiple
stationary distributions. Note also that, in general, $P^{\rm{seg}}(s)$ need
not coincide with $\Pstat(s)$, unless $\dict$ is a \balanced\ tree.} Letting
$\lenp^{\rm{seg}}(\dict)$ denote the expected dictionary length when the
distribution on the initial state is $P^{\rm{seg}}(s)$, \cite[Lemma
1]{tjalkens87} states that, for any state $t \in \alp^k$,
\begin{equation} \label{eq:TandW}
\sum_{s \in \alp^k} P^{\rm{seg}}(s) \Exp_{P,\dict} \;n_t (s,X^\ast) =
\Pstat (t) \lenp^{\rm{seg}}(\dict) \,,
\end{equation}
where $n_t (s,X^\ast)$ is the same as $n_t (X^\ast)$ but with the source
starting at state $s$, rather than $s_0$.\footnote{While the statement
of~\cite[Lemma 1]{tjalkens87} uses the above specific stationary
distribution, the proof in fact holds for any stationary distribution. It is
essentially based on counting the number $N_j(t)$ of occurrences of $t$ in a
sequence composed of $j$ source segments, for large $j$. Since the state
distribution at the segment starting points converges to $P^{\rm{seg}}(s)$,
with probability one, $N_j(t) = j \Pstat (t)
\lenp^{\rm{seg}}(\dict)$. On the other hand, for segments starting at
state $s$, $t$ occurs $\Exp_{P,\dict} \; n_t (s,X^\ast)$ times in the limit.}
Hence,
\begin{align}
&\sum_{s \in \alp^k} P^{\rm{seg}}(s) \Hdict[P,s] \nonumber\\
&\; = \sum_{t \in \alp^k}
H_P (X|t) \sum_{s \in \alp^k} P^{\rm{seg}}(s) \Exp_{P,\dict} \;
n_t (s,X^\ast) \nonumber \\
&\;= \sum_{t \in \alp^k} H_P (X|t) \Pstat (t) \lenp^{\rm{seg}}(\dict) =
\Hrate \lenp^{\rm{seg}}(\dict) \,, \label{eq:mymassey}
\end{align}
where the first equality follows from~(\ref{eq:leaf-entropy}), the second one
from~(\ref{eq:TandW}), and the third one from~(\ref{eq:Hrate}). Notice,
however, that the expected dictionary length may, in general, be quite
sensitive to the initial state. Therefore, it is not clear whether the
rightmost side of~(\ref{eq:mymassey}), which involves
$\lenp^{\rm{seg}}(\dict)$, can provide information on $\lenp(\dict)$. In
addition, a dictionary that satisfies the conditions of
Lemma~\ref{lem:asymptotics} for a given initial state $s_0$ may not satisfy
these conditions for a different initial state. While this issue is less
serious (as the class type size is not very sensitive to the initial state),
a direct application of the bounds shown in the lemma to the left-most side
of~(\ref{eq:mymassey}) is problematic.

Next, we present an auxiliary lemma that will address these problems. To
state the lemma, we need to introduce some additional objects. Given two
states $s,t \in \alp^k$, consider the set $\dict_{t,s}$ of all sequences
$\xs$ such that $s=s(t,\xs)$ and no proper prefix of $\xs$ has this property.
Since $\dict_{t,s}$ has the prefix property and since every state is
reachable from any other state, it is a (\complete, infinite) dictionary with
a failure probability that vanishes as the truncation level grows for any $P
\in \Pclassk$. The expected dictionary length for $\dict_{t,s}$ (where the
probabilities are computed with an assumed initial state $t$), which is the
mean first passage time from $t$ to $s$, is finite (since all states are
positive recurrent); we denote it by $\mathcal{L}_{t,s}$.

For sequence sets $U$ and $V$, let $U{\cdot}V$ denote the set $\{uv\,|\,u\in
U,\,v\in V\}$. If $U$ and $V$ are dictionaries, then $U{\cdot}V$ is a
dictionary whose corresponding tree is obtained by ``hanging'' the tree
corresponding to $V$ from each of the leaves of the tree corresponding to
$U$.

\begin{lemma}\label{lem:Massey}
Let $\{ \mathbb{D}_s \}_{s \in \alp^k}$ be a collection of sets of
dictionaries, one set for each state in $\alp^k$, which satisfies the
following property: For every pair of states $s,t$, if $\dict_s \in
\mathbb{D}_s$ then $\dict_{t,s} {\cdot} \dict_s \in \mathbb{D}_t$.
Let $L_s^*$ denote the infimum over $\mathbb{D}_s$ of the expected dictionary
length, where the expectation assumes the initial state $s$, $s \in \alp^k$,
and $P \in \Pclassk$ is arbitrary. Assume $L_s^*$ is finite and let
$\dict_s^* \in \mathbb{D}_s$, $s \in \alp^k$, denote a dictionary that
attains this infimum within some $\epsilon > 0$. Then, for any $s \in
\alp^k$, we have
\[
\min_{t \in \alp^k} \Hdic[P,t] (\dict_t^*) + K_1 \le \Hrate L_s^*
\le \max_{t \in \alp^k} \Hdic[P,t] (\dict_t^*) + K_2
\]
for some constants $K_1$ and $K_2$ that are independent of $\{ \mathbb{D}_s
\}$.
\end{lemma}

The proof of Lemma~\ref{lem:Massey} is presented in
Appendix~\ref{app:Massey}.

\subsection{Performance: Tight Bounds}
\label{sec:perbounds}
In view of Lemma~\ref{lem:condition-VF} and~(\ref{eq:lower-HD}) in Part~(i)
of Lemma~\ref{lem:asymptotics}, to obtain a lower bound on $\lenp (\dict)$
for universal \VFRs, it suffices to apply Lemma~\ref{lem:Massey} to a
suitable collection of sets $\{ \mathbb{D}_s \}_{s \in \alp^k}$; the lower
bound will translate to truncated dictionaries by typicality arguments
similar to those employed in the proof of Lemma~\ref{lem:asymptotics}. Our
lower bound is stated in Theorem~\ref{theo:lower} below.

\begin{theorem}\label{theo:lower}
Let $\vfr=(\dict,\xmap,M)$ be a universal \VFR\ such that
$P(\failn)\stackrel{N\to\infty}{\longrightarrow} 0$ for every $P \in
\Pclassk$. Then, for every $P \in \Pclassk$ and every $N >
(\log M)/(\Hrate-\epsilon)$, where $\epsilon > 0$ is arbitrary, we have
\begin{equation}\label{eq:lowerbound}
\lenpm(\dictn) \ge \frac{\log M + (K/2) \log \log M - O(1)}{\Hrate} \,.
\end{equation}
\end{theorem}

\begin{proof}
By Lemma~\ref{lem:condition-VF}, $|T(\xs)| \ge M$ for every $\xs \in \dict$.
Let $\dict'$ denote the dictionary obtained by ``pruning'' $\tree_\dict$ as
follows: Replace every $\xs \in \dict$ such that $|\xs| > N_1$ (where $N_1$
is given in~(\ref{eq:N1}) with $\epsilon$ an arbitrary positive constant) by
its shortest prefix of length at least $N_1$ whose type class contains at
least $M$ elements. We first prove that the expected length of $\dict'$
satisfies the claimed lower bound.

To this end, consider the collection of dictionary sets $\{ \mathbb{D}_s
\}_{s\in\alp^k}$, where $\mathbb{D}_s$ is the set of dictionaries
$\dict_s$ such that, for an initial state $s$, $|T(\xs)| \ge M$ for every
$\xs \in \dict_s$ and $P(\failn)\stackrel{N\to\infty}{\longrightarrow} 0$ for
every $P \in \Pclassk$. We show that this collection satisfies the conditions
of Lemma~\ref{lem:Massey}. It is easy to see that $L_s^*$ (where we use the
notation introduced in Lemma~\ref{lem:Massey}) is finite for all $s \in
\alp^k$ (in fact, by typicality arguments, $L_s^* \le N_1$) although, in any
case, the claimed lower bound on $\lenpm(\dict')$ would be trivial if $L_s^*$
were infinite. To see that if $\dict_s \in \mathbb{D}_s$ then $\dict_{t,s}
{\cdot} \dict_s \in \mathbb{D}_t$ for every pair of states $s,t$, it suffices
to observe that for a sequence $\xs = uv$ such that $u$ takes the source from
state $t$ to state $s$, if $v' \in T(v)$ (where the type class assumes an
initial state $s$) then $uv' \in T(\xs)$ (where the type class assumes an
initial state $t$). Therefore, if $v \in \dict_s$, we conclude that $|T(\xs)|
\ge M$ for $\xs \in \dict_{t,s} {\cdot} \dict_s$. Since, clearly,
$\dict_{t,s} {\cdot} \dict_s$ also has vanishing failure probability, the
collection $\{ \mathbb{D}_s \}_{s \in \alp^k}$ indeed satisfies the
conditions of Lemma~\ref{lem:Massey}.

Now, applying Part~(i) of Lemma~\ref{lem:asymptotics} to $\dict_t^*$, we
obtain
\[
\Hdic[P,t] (\dict_t^*) \ge \log M + (K/2) \log \log M - O(1)
\]
for every $t \in \alp^k$. It then follows from Lemma~\ref{lem:Massey} that
\[
L_{s_0}^* \ge \frac{\log M + (K/2) \log \log M - O(1)}{\Hrate} \,.
\]
Since $\dict' \in \mathbb{D}_{s_0}$ and $L_{s_0}^*$ is the infimum over
$\mathbb{D}_{s_0}$ of the expected dictionary length, the proof of our claim
on $\dict'$ is complete.

Next, we observe that since for every internal node $\xs$ at depth larger
than $N_1$ of the tree corresponding to $\dict'$ we have $|T(\xs)| < M$, then
\[
\lenpm(\dict')-\lenpm(\dictn[N_1]) \le \sum_{i=0}^\infty (i+1)
P(|T(x^{N_1+i})| < M) \,.
\]
Therefore, our claim on $\lenpm(\dictn)$ follows from our lower bound  on
$\lenpm(\dict')$ by the typicality arguments used in the proof of
Lemma~\ref{lem:asymptotics}, by which the probability that the type class of
a sequence be smaller than $2^{N_1(\Hrate-\epsilon)}$ decays exponentially
fast with the sequence length $N_1+i$.
\end{proof}

\begin{remark}\label{rem:surprise}
Although Theorem~\ref{theo:lower} is stated for a universal \VFR, it is easy
to see that it applies, like Lemma~\ref{lem:condition-VF}, also to perfect
\VFRs\ for ``almost all'' $P \in \Pclassk$. At first sight, this fact may
seem surprising, since perfect \VFRs\ for arbitrary memoryless distributions
$P\in\Pclassk[0]$ with expected dictionary length of the form $(\log M +
O(1))/\Hrate$ are described in~\cite{Roche91} and~\cite{HanHoshi97}, where it
is also shown that these \VFRs\ are optimal up to a constant term. However,
notice that, unlike the setting in Theorem~\ref{theo:lower}, these
\VFRs\ are not required to be perfect at all truncation levels. Thus, in the
memoryless case, the extra cost incurred by requiring perfection at all
truncation levels is at least $(\log\log M)/(2\Hrate)$. Since, as will be
shown in the sequel, the lower bound~(\ref{eq:lowerbound}) is achievable, the
above minimum value of the extra cost is also achievable. The situation in
the case $k>0$, for which~\cite{HanHoshi97} proposes a \VFR\ but does not
provide tight bounds, is discussed later.
\end{remark}

Next, we show that the lower bound~(\ref{eq:lowerbound}) is achievable. To
this end, we use Part~(ii) of Lemma~\ref{lem:asymptotics}, for which we need
a universal \VFR\ such that the size of the type class of each sequence in
its dictionary is at most $CM$ for some constant $C$, except for a negligible
subset of dictionary members. Such a bound on the type class size does not
appear to follow easily from the definition of the optimal \VFR\ $\vfrs$
since, in principle, the construction may require $|T|>CM$ for any constant
$C$ to guarantee $|\dict^*(T)| \ge M$. Therefore, we will use a different
universal \VFR\ to show achievability.

To describe this universal \VFR\ we will make use of an auxiliary dictionary
$\tilde{\dict}$, given by
\begin{align}
\tilde{\dict} = \left\{\,\xs \,\bigbar \, \right.&
|\typecut[T(\xs)]{u}| \ge M \;\, \forall u \in \alp^{k+1},\nonumber\\
&\left.\mbox{and no  } x^{\astast}\pprefix\xs \mbox{ has this property} \,
\medrule\right\}\,,
\label{eq:nochi}
\end{align}
where $\pprefix$ denotes the proper prefix relation. Thus, $\tilde{\dict}$
grows until the first time \emph{each} set $\typecut[T(\xs)]{u}$ (which, by
Lemma~\ref{lem:typesplit} and Remark~\ref{rem:distinct}, if not empty, are
distinct type classes in $\types[n-k-1]$) is large enough.

\begin{remark}
The ``stopping set'' $S$ defined in~\cite[Section IV]{Zhou-Bruck-arch} is the
special case of $\tilde{\dict}$ for the class of Bernoulli sources.
\end{remark}

We first show that Part~(ii) of Lemma~\ref{lem:asymptotics} is applicable to
$\tilde{\dict}$.

\begin{lemma}\label{lem:CM}
For every $P \in \Pclassk$, there exists a constant $C$ such that
$|T(\xs)|<CM$ for every $\xs \in \dictld \setminus\dictld_0$, where
$\dictld_0$ is a subset of $\dictld$  satisfying
\begin{equation} \label{eq:deltalike}
\sum_{\xs \in \dictld_0} |\xs| P(\xs) = O(1) \,.
\end{equation}
In addition, the probability of the failure set of $\dictld$ vanishes
(exponentially fast).
\end{lemma}

\begin{proof}
Recalling the definition~(\ref{eq:deltacounts}), let
\[
\dictld_0 \defined \dictld \cap \Bigl( \;\bigcup_{n \ge 1} S^{(\delta)}_n \;\Bigr )
\]
where $\delta>0$ is small enough for $\{ x^n \in S^{(\delta)}_n \}$ to be, by
our assumptions on $\Pclassk$, a large deviations event, so its probability
vanishes with $n$ exponentially fast. Therefore, (\ref{eq:deltalike}) holds.

Next, observe that, by Whittle's formula for $|T(x^n)|$, if $x^n \not\in
S^{(\delta)}_n$ then deleting its last symbol can only affect the type class
size by a multiplicative constant independent of $n$. As a result, for all
$u$, $|T(x^n)| < \beta |\typecut[T(x^n)]{u}|$ for some constant $\beta$ (that
depends on $\delta$ and $|u|$). Since, by~(\ref{eq:nochi}) and~(\ref{eq:au}),
if $\xs \in \tilde{\dict}$ then $|\typecut[T(\xs)]{u}|<M$ for some $u \in
\alp^{k+2}$, we conclude that for every sequence $\xs \in \dictld \setminus
\dictld_0$ we have $|T(\xs)| < CM$ for some constant $C$, as claimed.

Finally, notice that the failure set of $\dictld$ at truncation level $N$
consists of sequences $x^N$ such that $|\typecut[T(x^N)]{u}|<M$ for some $u
\in \alp^{k+1}$. The event $\{ x^N \in S^{(\delta)}_N \}$ has exponentially
vanishing probability. If $x^N \not\in S^{(\delta)}_N$ then, as discussed,
$|T(x^N)| < \beta |\typecut[T(x^N)]{u}|$. Hence, clearly, $P(T(x^N)) < \beta
P(\typecut[T(x^N)]{u})$. If $N > N_1$ (as defined in~(\ref{eq:N1})),
$P(|\typecut[T(x^N)]{u}| < M)$ vanishes exponentially fast. Therefore, so
does the probability of the failure set.
\end{proof}

We can now use the upper bound~(\ref{eq:HD}) and Lemma~\ref{lem:Massey} to
upper-bound $\lenpm(\tilde{\dict})$ as follows.

\begin{lemma}\label{lem:lendictilde}
For every $P \in \Pclassk$ we have
\begin{equation}\label{eq:upperbound}
\lenpm(\tilde{\dict}) \le \frac{\log M + (K/2) \log \log M + O(1)}{\Hrate} \,.
\end{equation}
\end{lemma}

\begin{proof}
Consider the collection $\{ \mathbb{D}_s \}_{s \in \alp^k}$ where
$\mathbb{D}_s$ is the set of dictionaries $\dict_s$ such that, for an initial
state $s$, if $x^\ast \in \dict_s$ then $|\typecut[T(\xs)]{u}| \ge M$ for all
$u \in \alp^{k+1}$. Clearly, the same arguments as in the proof of
Theorem~\ref{theo:lower} prove that the collection $\{ \mathbb{D}_s \}_{s \in
\alp^k}$ satisfies the conditions of Lemma~\ref{lem:Massey}.
By~(\ref{eq:nochi}), the dictionary with shortest expected length over
$\mathbb{D}_s$ is precisely $\tilde{\dict}$ with initial state $s$, which we
denote $\tilde{\dict}_s$. Therefore, by Lemma~\ref{lem:Massey},
\[
\Hrate \lenpm(\tilde{\dict}_s) \le \max_{t \in \alp^k}
\Hdic[P,t] (\tilde{\dict}_t)  + K_2
\]
for some constant $K_2$. Now, by Lemma~\ref{lem:CM} and Part~(ii) of
Lemma~\ref{lem:asymptotics}, we have
\[
\Hdic[P,t] (\tilde{\dict}_t) \le \log M + (K/2) \log \log M + O(1)
\]
independently of the initial state, which completes the proof.
\end{proof}

Thus, the expected length of $\tilde{\dict}$ coincides (up to an additive
constant) with the lower bound of Theorem~\ref{theo:lower} on the expected
length of a universal \VFR. Recall that Lemma~\ref{lem:condition-VF} requires
that, for the dictionary $\dict$ of a universal \VFR, $M$ divide $|\dict
(T(\xs))|$ for all $\xs \in \dict$. While $\dictld$ may not satisfy this
property, the following result makes it a suitable ``building block'' in the
construction of a universal \VFR.

\begin{lemma}\label{lem:availability}
For every $\xs \in \tilde{\dict}$ we have $|\tilde{\dict} (T(\xs))| \ge M$.
\end{lemma}

\begin{proof}
Let $x^n \in \tilde{\dict}$, $T \defined T(x^n)$, and $T' \defined
T(x^{n-1})$. For a generic string $w \in \alp^\ast$, consider the following
properties:
\begin{itemize}
\itemsep=-0.2mm
\item [(P1)] $|\typecut[T']{w}| \ge M$;
\item [(P2)] For every suffix $t$ of $w$, there exists $u \in \alp^{k+1}$
    such that $|\typecut[T']{ut}|<M$;
\item [(P3)] $|\typecut[T']{aw}| < M$ for all $a \in \alp$.
\end{itemize}
We claim the existence of a string $w$ satisfying~(P1)--(P3). We will exhibit
such a string,  by constructing a sequence $v^{(0)}, v^{(1)}, \cdots v^{(i)},
\cdots$ of strings of strictly increasing length, each satisfying~(P1)--(P2),
and with the property that, given $v^{(i)}$, $i\ge 0$, there exists a string
$ z $ such that either $w =  z  v^{(i)}$ satisfies~(P1)--(P3) (and our claim
is proven) or, for some $c\in\alp$, $v^{(i+1)}=c z  v^{(i)}$ satisfies
(P1)--(P2), and we can extend the sequence by one element. Such a
construction cannot continue indefinitely without finding the desired string
$w$, since, as the length of $v^{(i)}$ increases, eventually we would have
$\typecut[T']{v^{(i)}} = \emptyset$, so $v^{(i)}$ would not satisfy (P1).

To construct the sequence $\{ v^{(i)} \}$, we first establish that
$v^{(0)}=\lambda$ satisfies (P1)--(P2). Since $T' = \typecut[T]{x_{n-k}}$ and
$x^n \in \dictld$, by~(\ref{eq:nochi}), $|T'|\ge M$, so (P1) holds for
$w=\lambda$. But since $x^{n-1} \notin \dictld$, again by~(\ref{eq:nochi}),
(P2) also holds for $\lambda$. Next, to prove the existence of the mentioned
string $ z  $ given $v^{(i)}$, $i\ge 0$, we need the following lemma, which
is proved in Appendix~\ref{app:lema14}.

\begin{lemma} \label{lem:small-ancestor}
Let $T_1 \in \types[n+1]$, $b,c \in \alp$, $T_2 = \typecut[T_1]{b}$, and $T_3
= \typecut[T_1]{c}$. Then, there exists $u \in \alp^{k+1}$ such that $|T_2|
\ge |\typecut[T_3]{u}|$.
\end{lemma}

Assume that we are given $v^{(i)}$, $i \ge 0$, satisfying (P1)--(P2). By
(P2), there exists $z' \in \alp^{k+1}$ such that $|\typecut[T']{z'
v^{(i)}}|<M$. \emph{A fortiori}, every suffix of $z' v^{(i)}$ also
satisfies~(P2). Let $z''$ denote the shortest suffix of $z'$ such that
$|\typecut[T']{z'' v^{(i)}}| < M$. Since $v^{(i)}$ satisfies~(P1), we have
$|z''|> 0$, so let $z''=b z  $, $b \in \alp$. Now, if for every $c \in \alp$
we have $|\typecut[T']{c z  v^{(i)}}| <  M$, then~(P1)--(P3) hold for $w= z
v^{(i)}$, as claimed. Otherwise, let $c$ be such that $|\typecut[T']{c z
v^{(i)}}| \ge  M$. By Lemma~\ref{lem:small-ancestor} (with $T_1 =
\typecut[T']{ z   v^{(i)}}$), there exists $r\in\alp^{k+1}$ such that
$|\typecut[T']{rc z  v^{(i)}}| \le |\typecut[T']{b z  v^{(i)}}| < M$.
Hence,~(P1)--(P2) hold for $v^{(i+1)}=c z  v^{(i)}$. We have thus shown the
existence of a string $w$ satisfying~(P1)--(P3).

Next, let $y^n {\in} \typecutnot[T]{vx_{n-k}}$, and denote $|w|{=}\ell$.
By~(\ref{eq:au}), $y^{n-\ell-1} {\in} \typecut[T']{v}$. Since $w$
satisfies~(P3), no prefix of $y^{n-\ell-1}$ satisfies the membership
condition of~(\ref{eq:nochi}). Consider now $y^{n-j-1}$, $0 \le j < \ell$.
Since $y^{n-1} {\in} \typecutnot[T']{w}$, it is easy to see that $y^{n-j-1}
{\in} \typecut[T']{t}$, where $t$ is the (proper) suffix of $w$ of length
$j$. Thus, by~(P2), $y^{n-j-1}$ does not satisfy the membership condition
either. It follows that the membership condition is not satisfied for any
proper prefix of $y^n$. But since $\typecutnot[T]{wx_{n-k}} \subseteq T$, the
condition is satisfied for $y^n$ and, hence, $y^n \in \dictld$. The proof is
complete by noticing that $|\typecutnot[T]{wx_{n-k}}| {=} |\typecut[T']{w}|$
and invoking~(P1) for $w$.
\end{proof}

Next, we describe the construction of the dictionary $\dictlds$ of a
universal
\VFR. In this construction, the value of the initial state implicit in the
class type definition in $\dictld$ will \emph{not} be fixed at the same $s_0$
throughout, and the value of the threshold for the class type size used
in~(\ref{eq:nochi}) may differ from $M$. Therefore, it will be convenient to
explicitly denote $\tilde{\dict}$ by $\tilde{\dict}_s(\ell)$, where $s$
denotes the initial state and $\ell$ is the threshold, while maintaining the
shorthand notation $\tilde{\dict} \defined \tilde{\dict}_{s_0}(M)$.

\begin{figure}
\begin{figprocedurenoIO}
\small
\item\label{step:dictld0}
Set $i=0$, $\dictlds = \emptyset$, and
let $\dictldi[0] = \dictld$. For $\xs\in\dictld$, define $\chunki[0](\xs) =
\dictld(\type[\xs])$.
\item\label{step:dictldi} Set $\dictldi[i+1]=\emptyset$.
 For each $\chunk = \chunki(x^\ast)$, $x^{\ast}\in\dictldi$, do:

\begin{enumerate}
\item\label{step:todict} Let $m = |\chunk|
\bmod M$, and let $U$ be a set of $m$ sequences from $\chunk$.\ \
Add $\chunk\setminus U$ to $\dictlds$.
\item \label{step:append} If $m>0$, let $s_f$ be the common final state
of all sequences in $U$.\ \ Add $U{\cdot}\dictld_{s_f}\left(\lceil
 M/m\rceil\right)$ to $\dictldi[i+1]$.
\end{enumerate}
\item\label{step:chunk i+1} If $\dictldi[i+1]=\emptyset$, \textbf{stop}.
Otherwise, for each $\xs \in \dictldi[i+1]$, let $x^{\ell}$ be its prefix in
$\dictldi$, and define

$
\quad\chunki[i+1](\xs) =
\{\,y^\ast\in\type[\xs]\cap\dictldi[i+1] \;|\; y^{\ell}\in\chunki(x^{\ell})\,\}\,.
$

Increment $i$, and go to Step~\ref{step:dictldi}.
\end{figprocedurenoIO}
\vspace{-0.55em}
\caption{\label{fig:dictlds}Description of the universal \VFR\ $\dictlds$.}

\vspace{-0.8em}
\end{figure}
The dictionary $\dictlds$ is iteratively described, as shown in
Figure~\ref{fig:dictlds}. At the beginning of the $i$th iteration, $\dictlds$
contains a prefix set of sequences which have been added to the set in
previous iterations. In addition, there is a prefix set $\dictldi$ of
sequences that are still pending processing (i.e., either inclusion in
$\dictlds$, or extension), such that $\dictlds \cup \dictldi$ is a
\complete\ dictionary. Initially, $\dictlds$ is empty and $\dictldi[0] =
\dictld$. Sequences in $\dictldi$ are collected into groups $\chunki(\xs)$,
where the latter consists of all the pending sequences of the same type as
$\xs$, and whose prefixes in the previous iteration were in the same group
$\chunki[i-1](\xs)$ (thus, recursively, the prefixes were also of the same
type in all prior iterations).

The dictionary $\dictlds$ is built up, in Step~\ref{step:todict}, of sets of
sizes divisible by $M$, consisting of sequences of the same type. Thus, $M$
divides $|\tilde{\dict}^\ast(T)|$ for all $n$ and all type classes
$T\in\types$, so that Lemma~\ref{lem:condition-VF} guarantees the existence
of a universal \VFR\ based on $\dictlds$. The remaining $m$ sequences are
recursively extended by ``hanging,'' in Step~\ref{step:append}, dictionaries
$\dictld_{s_f} \left(\lceil M/m\rceil\right)$. Thus, by
Lemma~\ref{lem:availability}, the new set $\Delta_{i+1}$ contains $m$ copies
of type classes of sizes at least $M/m$. Unless, at some step $i$, $m=0$ for
all groups (i.e., $\dictldi[i+1]$ is empty), the procedure continues
indefinitely and, as a result, $\tilde{\dict}^\ast$ contains more infinite
paths than $\tilde{\dict}$. The choice of threshold $\lceil M/m\rceil$ in
Step~\ref{step:append} guarantees the following property for the groups
$\chunki$ in Step~\ref{step:todict}.

\begin{lemma}\label{lem:atleastM}
For all $i\ge 0$ and all $\xs\in\dictldi$, we have $|\chunki(\xs)|\ge M$.
\end{lemma}

\begin{proof}
Lemma~\ref{lem:availability} guarantees that the claim is true for $i=0$
(Step~\ref{step:dictld0} in Fig.~\ref{fig:dictlds}). By
Steps~\ref{step:dictld0} and~\ref{step:chunk i+1}, sequences in the same
group $\chunki$ are indeed of the same type, and, thus, $s_f$ is well defined
in Step~\ref{step:append}, where $\dictldi[i+1]$ is built-up of subsets of
the form $U{\cdot}\dictld_{s_f}(\lceil M/m\rceil)$. Also, by
Lemma~\ref{lem:availability}, the size of the type class of every sequence in
$\dictld_{s_f}(\lceil M/m\rceil)$ is at least $\lceil M/m\rceil$. Now,
sequences in the same group $\chunki$, ending in state $s_f$, when appended
with sequences from $\dictld_{s_f}(\lceil M/m\rceil)$ that are of the same
type with respect to the \emph{initial} state $s_f$, remain in the same group
$\chunki[i+1]$ (Step~\ref{step:append} and definition of $\chunki[i+1]$ in
Step~\ref{step:chunk i+1}). In particular, this applies to the sequences in
the set $U$, and, therefore, $|\chunki[i+1](\xs)|\ge m\lceil M/m\rceil\ge M$
for all $\xs\in\dictldi[i+1]$.
\end{proof}

We now turn to the computation of $\lenpm(\tilde{\dict}^\ast)$. We first
show, in Lemma~\ref{lem:constant} below, that the iterative process described
in Figure~\ref{fig:dictlds} does not increase the expected length of
$\tilde{\dict}$ by more than a constant.

\begin{lemma}\label{lem:constant}
For every $P \in \Pclassk$ we have
\[
\lenpm(\tilde{\dict}^\ast) - \lenpm(\tilde{\dict}) = O(1) \,.
\]
\end{lemma}

\begin{proof}
Let $\mathbb{G}_i$ denote the partition of $\dictldi$ into groups $\chunki$,
and for each $\chunk \in \mathbb{G}_i$, let $m(\chunk) = |\chunk| \bmod M$.
Let $\mathcal{L}_i$ denote the limiting expected truncated length of the
\complete, prefix set formed by the union of $\dictldi$ and the ``current
state'' of $\tilde{\dict}^\ast$ after the $i$th iteration of the algorithm.
Clearly,
\begin{equation}\label{eq:limit}
\lenpm(\tilde{\dict}^\ast) = \lim_{i \to \infty} \mathcal{L}_i \,.
\end{equation}
Since all the sequences in a group $\chunk$ are equiprobable and only
$m(\chunk)$ of them are extended in Step~\ref{step:append} whereas, by
Lemma~\ref{lem:atleastM}, at least $M$ of them are not, we have
\begin{equation}\label{eq:extlength}
\mathcal{L}_{i+1} \le \mathcal{L}_i + \sum_{\chunk \in \mathbb{G}_{i+1}}
\frac{m(\chunk)}{M+m(\chunk)} P(\chunk) \lenpm
\left ( \dictld_{s_f} \left(\lceil M/m(\chunk)\rceil\right) \right )
\end{equation}
where $s_f$ denotes the common final state of the sequences in $\chunk$. By
Lemma~\ref{lem:lendictilde}, applied rather loosely, for every $s \in \alp^k$
and every positive integer $\ell$, we have $\lenpm(\dictld_s (\ell)) < \eta
\ell$ for some constant $\eta$, independent of $M$. Therefore,
(\ref{eq:extlength}) implies
\begin{equation}\label{eq:gamma}
\mathcal{L}_{i+1} \le \mathcal{L}_i + \eta \sum_{\chunk \in \mathbb{G}_{i+1}}
P(\chunk) = \mathcal{L}_i + \eta P(\dictldi[i+1]) \,.
\end{equation}
Now, since $m(\chunk)<M$, it follows from Lemma~\ref{lem:atleastM} that more
than half of the pending sequences make it to $\dictlds$ in
Step~\ref{step:todict}. Thus, since $\dictldi$ is a prefix set, we have
$P(\dictldi) > 2P(\dictldi[i+1])$ and, hence,
\[
P(\dictldi[i+1]) < 2^{-(i+1)} \,.
\]
It then follows from~(\ref{eq:gamma}) that
\[
\mathcal{L}_i < \mathcal{L}_0 + 2 \eta = \lenpm(\tilde{\dict}) + 2 \eta
\]
which, together with~(\ref{eq:limit}), completes the proof.
\end{proof}

Therefore, just as $\tilde{\dict}$, $\tilde{\dict}^\ast$ attains the lower
bound~(\ref{eq:lowerbound}). The following characterization of
$\lenpm(\dictns)$ then follows straightforwardly from
Theorem~\ref{theo:lower}, Lemma~\ref{lem:lendictilde},
Lemma~\ref{lem:constant}, and Theorem~\ref{theo:optimal}.

\begin{theorem}\label{theo:achievability}
For every $P \in \Pclassk$ and sufficiently large $N$, the truncated
dictionary $\dictns$ of the optimal universal \TVFR\ $\vfrns$ satisfies
\begin{equation} \label{eq:optlength}
\lenpm(\dictns) = \frac{\log M + (K/2) \log \log M + O(1)}{\Hrate} \,.
\end{equation}
\end{theorem}

\begin{remark}
The term $(K/2) \log \log M$ in~(\ref{eq:optlength}) resembles a typical
``model cost'' term in universal lossless compression but, as mentioned in
Remark~\ref{rem:surprise}, the universality of the \VFR\ does not entail an
extra cost. Instead, in the memoryless case, $(\log\log M)/(2\Hrate)$ is, as
follows from the results in~\cite{Roche91} and~\cite{HanHoshi97} discussed in
Remark~\ref{rem:surprise}, the extra cost of maintaining perfection under
truncation, in either the universal or individual process cases. The case
$k>0$ is also discussed in~\cite{HanHoshi97}, but the bounds provided in that
work are not tight enough to reach a similar conclusion. Nevertheless,
Lemma~\ref{lem:Massey} yields a tighter lower bound on the expected
dictionary length for any perfect \VFR\ without the truncation requirement.
To derive this bound, let $\mathbb{D}_s$ be the set of dictionaries $\dict_s$
corresponding to such \VFRs\ for an initial state $s \in \alp^k$. Clearly,
$\{ \mathbb{D}_s \}_{s \in \alp^k}$ satisfies the conditions of
Lemma~\ref{lem:Massey}. Since the dictionary strings can be clustered into
$M$ groups, each of probability $1/M$ (in the limit), we have $\Hdic[P,s]
(\dict_s) \ge \log M$ for any $\dict_s \in \mathbb{D}_s$ and any $s \in
\alp^k$. Thus, the lemma implies that $(\log M-O(1))/\Hrate$ is a lower bound
on the expected dictionary length for any perfect \VFR\ (without the
truncation requirement). We conclude that, for $k>0$, the extra cost of
maintaining perfection under truncation, in either the universal or
individual process cases, is \emph{at most} $(K/2) \log \log M$. The question
of whether this value is also a lower bound remains open, as it requires to
improve on the upper bound provided in~\cite{HanHoshi97} on the expected
dictionary length of the ``interval algorithm.'' The second order asymptotic
analysis of the performance of universal \VFRs\ on which no truncation
requirements are posed also remains an open question.
\end{remark}

\appendices
\section{Proof of Lemma~\ref{lem:condition-fv}}\label{app:condition-fv}

Before we proceed with the proof, we define the subset $\pDomain_0$ mentioned
in the statement of the lemma. Consider functions
$g(\bldp)=\sum_{T\in\types}g_T P(T)/|T|$, where the $g_T$ are integers, and
$P(T)$ is the total probability of the type class $T$ for a parameter $\bldp
\in \pDomain$. These functions are multivariate polynomials in the components
of $\bldp$. Let
\begin{equation}\label{eq:Gn}
G_n = \left\{\, g(\bldp) \bigbar[0.9em]
|g_T|{\le}|T|\;\,\forall T{\in}\types,\;g_T\ne 0\;\text{for some}\;T{\in}\types\right\}.
\end{equation}
It is known (see, e.g.,~\cite{MerhavWeinb04,MSW-delay-limited08}) that the
type probabilities $P(T)$, as functions of $\bldp$, are linearly independent
over the reals. Thus, no $g\in G_n$ is identically zero. Let $\pDomain_0$
denote the set of all vectors $\bldp$ such that $g(\bldp)= 0$ for some $g \in
G_n$. It is readily verified that $\pDomain_0$ has volume zero in $\pDomain$.

\begin{proof}[Proof of Lemma~\ref{lem:condition-fv}]
Let $\Pbldp$ be a process in $\Pclassk$, where we use $\Pbldp$ instead of $P$
to emphasize the dependence of the probabilities on the parameter vector
$\bldp\in\pDomain$. Consider a pair $(r,M)\in\posints\times\target$ such that
$r\in\setM$ and $\probM \defined \Pbldp\bigl(\Mmap(x^n)=M\bigr)\ne 0$, and
let $\rMinvx$ denote the set of sequences $x^n$ such that $\Mmap(x^n)=M$ and
$\Rmap(x^n)=\rnd$. Since sequences of the same type are equiprobable, we have
\begin{align}
\Pbldp & \Bigl(\Rmap(X^n){=}\rnd \,\bigl|\, \Mmap(X^n){=}M \bigr. \Bigr) =
\probM^{-1}\sum_{x^n\in\rMinvx}P_{\bldp}( x^n)\nonumber\\
&= \probM^{-1}\sum_{\typet\in\types[n]}\, \sum_{x^n\in\rMinvx\cap\typet}P(x^n)
= \probM^{-1}\sum_{T\in\types[n]} \frac{ |\rMinvx \cap T|}{|T|}\Pbldp (T)\,.
\label{eq:linind}
\end{align}
If the condition of the lemma holds, then the right-hand side
of~(\ref{eq:linind}) is independent of $\rnd$ and, thus, $\fvr$ is universal.
Conversely, if $\fvr$ is perfect for some $\bldp\in\pDomain$, it follows
from~(\ref{eq:linind}) that for any $\rnd,\rnd' \in \setM$, we have
\begin{equation} \label{eq:lindif}
\sum_{T\in\types[n]} \frac{ |\rMinvx \cap T| - |\rMinvx' \cap T|}{|T|}P_{\bldp}(T) = 0\,,
\end{equation}
where $\rMinvx'$ is defined as $\rMinvx$, but for $\rnd'$. If the condition
of the lemma does not hold, then, for some choice of $M$,
$\rnd,\rnd'\in\setM$, and $T$, we have $|\rMinvx \cap T| \ne |\rMinvx' \cap
T|$, and the expression on the left-hand side of~(\ref{eq:lindif}), viewed as
a multivariate polynomial in the components of $\bldp$, belongs to $G_n$.
Thus, by the definition of $\pDomain_0$, we must have $\bldp \in \pDomain_0$.
\end{proof}

\section{Proof of Lemma~\ref{lem:chcinv}}\label{app:chcinv}
\begin{proof}
We prove the lemma by induction on $m$. For $m=1$, we have $H=0$ and the
claim is trivial. For $m>1$, define $\bldq' =
[q_1',\,q_2',\,\ldots,\,q_{m-1}']$ with $q_i' = q_{i+1}/(1-q_1)$, $1\le i <
m$. We claim that $\bldq'$ satisfies the version of~(\ref{eq:qcond}) for
distributions on $m-1$ symbols. Indeed, for $1\le i < m$,  we have
\[
c\,q_i' = \frac{c\,q_{i+1}}{1-q_1}
\ge \frac{1}{1-q_1}\Bigl( 1 - \sum_{j=1}^i q_j \Bigr) =
1 - \sum_{j=1}^{i-1} q_j'\,,
\]
where the first inequality follows from the assumptions of the lemma, and the
last equality from the definition of $\bldq'$ and some algebraic
manipulations. Now, denoting by $H'$ the entropy of $\bldq'$, we can write
\begin{align*}
H &= h(q_1) + (1-q_1)H'
 \,\le\, h(q_1) + (1-q_1)c\,h(c^{-1})\\
&= q_1\left(\frac{h(q_1)}{q_1} - \frac{h(c^{-1})}{c^{-1}}\right)
+ c\,h(c^{-1}) \,\le\, c\,h(c^{-1})\,,
\end{align*}
where the first inequality follows from the induction hypothesis, and the
second inequality from the fact that, by~(\ref{eq:qcond}), we have $q_1\ge
c^{-1}$ and the function $h(x)/x$ is monotonically decreasing.
\end{proof}

\section{Proof of Lemma~\ref{lem:aux}}
\label{app:aux}

\begin{proof}[Proof of Lemma~\ref{lem:aux}]
The proof is by induction on $|\mathcal{B}|$. The claim is trivial for
$|\mathcal{B}|=1$. Assume it also holds for all $|\mathcal{B}| < B$, where $B
> 1$. Then, if $|\mathcal{B}| = B$, letting $u_0$ denote an element with
maximum probability in $\mathcal{B}$, we have
\begin{align*}
\sum_{u,v \in \mathcal{B}}& | R(u){-}R(v) | \\
&=\!\!\!\!\sum_{u,v \in \mathcal{B} \setminus \{u_0\}} | R(u){-}R(v) | +
2\!\! \sum_{v \in \mathcal{B} \setminus \{u_0\}} ( R(u_0){-}R(v) ) \\
&\le 2 (B-2) \left(\medrule R(\mathcal{B}){-}R(u_0)\right)
   + 2 \left(\medrule B R(u_0){-}R(\mathcal{B})\right) \\
&= 4 R(u_0) + 2 (B-3) R(\mathcal{B}) \le 2(B-1) R(\mathcal{B})\,,
\end{align*}
where the first inequality follows from the induction hypothesis and the last
one from $R(u_0) \le R(\mathcal{B})$.
\end{proof}

\section{Proof of Lemma~\ref{lem:condition-VF}}
\label{app:condition-vf}

\begin{proof}[Proof of Lemma~\ref{lem:condition-VF}]
Consider the \VFR\ $\vfr = (\dict,\xmap,M)$. For a type class
$\typet\in\types$, and $r\in\fixalp$, define the set $\dict(T)_r = \left\{
x^n\in\dict(T)\,|\,\xmap(x^n) = r \right\}$. Assume that $\vfr$ satisfies the
condition of the lemma, i.e., that $\left|\dict(T)_r\right|$ is independent
of $r$ for all $n$ and all $T\in\types$. We claim that $\vfrn$ is universal
at all truncation levels $N$ (and, thus, $\vfr$ is universal). Indeed,
letting, for conciseness, $\Pbldp(r,N)$ denote the probability
$\Pbldp\big(\xmap(\Xs) = r ,\,\Xs\in\dictn\big)$, where, again, we emphasize
the dependence of a process $\Pbldp\in\Pclassk$ on its parameter $\bldp$, we
have
\begin{align}
\Pbldp&(r,N) =
\sum_{n=1}^N \sum_{\stacked{x^n\in\dictn,}{\xmap(x^n)=r}}\Pbldp(x^n)\nonumber\\
& =
\sum_{n=1}^N \,\sum_{T\in\types}\,\sum_{x^n\in\dict(T)_r}\!\!\!\Pbldp(x^n)
=
\sum_{n=1}^N \,\sum_{T\in\types}\frac{\left|\dict(T)_r\right|}{|T|}{\Pbldp(T)}\,.
\label{eq:Prob(r)}
\end{align}
The expression on the right-hand side of~(\ref{eq:Prob(r)}) is independent of
$r$, establishing our claim. Assume now that $\vfr$ does not satisfy the
condition of the lemma, and let $N$ be the smallest integer for which a type
$T'\in\types[N]$ violates the condition, i.e., for some $r,r'\in\fixalp$, we
have $|\dictn(T')_r|\ne|\dictn(T')_{r'}|$. Consider a process $\Pbldp$  such
that $\vfrn$ is perfect for $\Pbldp$. By~(\ref{eq:Prob(r)}), applied also to
$r'$, we have
\begin{align}
\Pbldp(r,N)&-\Pbldp(r',N) \nonumber\\
&=
\sum_{n=1}^N \,\sum_{T\in\types}
\frac{\left|\dict(T)_r\right|-\left|\dict(T)_{r'}\right|}{|T|}{\Pbldp(T)}
\nonumber\\
&=\sum_{T\in\types[N]}\frac{\left|\dict(T)_r\right|-\left|\dict(T)_{r'}\right|}{|T|}
{\Pbldp(T)}\,,
\label{eq:P(r)-P(r')}
\end{align}
where the last equality follows from our assumption on $N$. Also from the
same assumption, it follows that at least one of the numerators in the
expression on the rightmost side of~(\ref{eq:P(r)-P(r')}) is nonzero, and the
expression, as a multivariate polynomial in the entries of $\bldp$, belongs
to $G_N$ as defined in~(\ref{eq:Gn}). We now reason as in the proof of
Lemma~\ref{lem:condition-fv} to conclude that $\bldp$ must belong to a
(fixed) subset $\pDomain_0'$ of measure zero in $\pDomain$.
\end{proof}

\section{Proof of Lemma~\ref{lem:Massey}}
\label{app:Massey}
\begin{proof}[Proof of Lemma~\ref{lem:Massey}]
We will denote by $L_{P,s}$ the expected dictionary length under $P$, with
initial state $s$. For every pair of states $t,s \in \alp^k$, since
$\dict_{t,s} {\cdot} \dict_s^* \in \mathbb{D}_t$, we have
\[
L_t^* \le L_{P,t} (\dict_{t,s} {\cdot} \dict_s^*) = \mathcal{L}_{t,s} +
L_{P,s} (\dict_s^*) \le L_s^* + \epsilon + \mathcal{L}_{t,s} \,.
\]
Therefore,
\begin{equation} \label{eq:close1}
|L_s^* - L_t^*| \le \max_{u,v \in \alp^k} \mathcal{L}_{u,v} + \epsilon \,.
\end{equation}
Now, let $\dict_{s,N}^*$ denote the truncation of $\dict_s^*$ to depth $N$
(completed with the corresponding failure set), and let $\dict^{\rm{univ}}_N$
denote the dictionary given by $\{ s u_s \,| \, s \in \alp^k ,\, u_s \in
\dict_{s,N}^* \}$. Thus, $\tree_{\dict^{\rm{univ}}_N}$ is obtained by taking
a \balanced\ tree of depth $k$, and ``hanging'' from each leaf $s$ the tree
$\tree_{\dict_{s,N}^*}$, $s \in \alp^k$. Clearly,
\begin{equation} \label{eq:univ}
L_{P,s} (\dict^{\rm{univ}}_N) = k + \sum_{t \in \alp^k}
P^{(k)}(t|s) L_{P,t} (\dict_{t,N}^*) \,,
\end{equation}
where $P^{(k)}(t|s)$ denotes the probability of moving from state $s$ to
state $t$ in $k$ steps. Since, for large enough $N$, $L_{P,t}
(\dict_{t,N}^*)$ is arbitrarily close to $L_t^*$, (\ref{eq:close1})
and~(\ref{eq:univ}) imply that, for every $s,t \in \alp^k$,
\begin{equation} \label{eq:close2}
|L_{P,s}(\dict^{\rm{univ}}_N) - L_t^*| < C_1 \,,
\end{equation}
where the constant $C_1$ depends on $P$ but is independent of $s$ and $t$.
Similarly,
\[
\Hdic[P,s] (\dict^{\rm{univ}}_N)  = H_{P,s} (X^k) + \sum_{t \in \alp^k}
P^{(k)}(t|s) \Hdic[P,t] (\dict_{t,N}^*)
\]
where $H_{P,s}(X^k)$ denotes the entropy of $k$-tuples (starting at state
$s$). Therefore, applying~(\ref{eq:mymassey}) to $\dict^{\rm{univ}}_N$, we
obtain
\begin{align}
\Hrate & \lenp^{\rm{seg}}(\dict^{\rm{univ}}_N)\nonumber\\
& = C_2 +
\sum_{s,t \in \alp^k} P^{\rm{seg}}(s) P^{(k)}(t|s) \Hdic[P,t](\dict_{t,N}^*)
\label{eq:end}
\end{align}
where $C_2$ is a constant that depends only on $P$. Together
with~(\ref{eq:close2}), and taking the limit as $N {\to} \infty$ so that
$\Hdic[P,t](\dict_{t,N}^*) {\to} \Hdic[P,t](\dict_t^*)$ (since the failure
probability vanishes), (\ref{eq:end}) implies the claim of the lemma.
\end{proof}

\section{Proof of Lemma~\ref{lem:small-ancestor}}
\label{app:lema14}
\begin{proof}[Proof of Lemma~\ref{lem:small-ancestor}]
Let $z_1^k$ denote the final state of $T_1$. We prove the lemma with $u = b
z_1^k$. If $\typecut[T_3]{u}$ is empty, the lemma holds trivially. Otherwise,
the sequences in $T_1$ include at least one transition from state $b
z_1^{k-1}$ to state $z_1^k$. It is easy to see that there exists a sequence
in $T_1$ such that one of these transitions is the final one, and therefore
$T_2$ is not empty. Thus, we assume that $T_2,T_3 \in\types$ and
$\typecut[T_3]{u} \in \types[n-k-1]$. The counts defining $\typecut[T_3]{u}$
differ from those defining $T_1$ in that a chain of state transitions
\[
b z_1^{k-1}{\rightarrow}z_1^k{\rightarrow\,}z_2^k c {\,\rightarrow\,}z_3^k c z_1
{\rightarrow\,}\cdots{\rightarrow\,}z_k c z_1^{k-2} {\rightarrow\,}c z_1^{k-1}
{\rightarrow}z_1^k
\]
has been deleted. On the other hand, $T_2$ is obtained by deleting from $T_1$
a transition $b z_1^{k-1} \rightarrow z_1^k$. Therefore, $\typecut[T_3]{u}$
can be obtained from $T_2$ by deleting a \emph{circuit} of state transitions
\begin{equation} \label{eq:deletions}
z_1^k \rightarrow z_2^k c \rightarrow z_3^k c z_1 \rightarrow \cdots
\rightarrow z_k c z_1^{k-2} \rightarrow c z_1^{k-1} \rightarrow z_1^k \,.
\end{equation}
At least one of the states, $t_1^k$, in the circuit~(\ref{eq:deletions}),
must occur in $x^{n-k-1} \in \typecut[T_3]{u}$, for otherwise the transition
graph of $T_2$, obtained by adding the circuit, would have a disconnected
component. Fix $x^{n-k-1} \in \typecut[T_3]{u}$, and let $i$ be such that
$x_{i+1}^{i+k} = t_1^k$ is the last occurrence of this state in the sequence,
$0 \le i \le n-2k-1$, and let $v_1^{k+1}$ denote the string of symbols
determined by the circuit~(\ref{eq:deletions}) starting at a transition from
state $t_1^k$ (and ending at the same state). Clearly, the sequence $x^{i+k}
v_1^{k+1} x_{i+k+1}^{n-k-1}$ is in $T_2$, as the inserted string generates
all the missing transitions prescribed by~(\ref{eq:deletions}), returning to
state $t_1^k$. In addition, it is easy to see that, with this procedure, two
different sequences in $\typecut[T_3]{u}$ generate two different sequences in
$T_2$, which completes the proof.
\end{proof}

\bibliographystyle{IEEEtran}

\end{document}